\nonstopmode
\documentclass{LMCS}

\def\doi{8 (1:29) 2012}                                                  
\lmcsheading%
{\doi}                                                                      
{1--36}                                                                    
{}                                                                         
{}                                                                        
{Oct.~\phantom04, 2011}                                                    
{Mar.~26, 2012}                                                            
{}                                                                              

\usepackage{xspace}
\usepackage{ifthen}
\usepackage{amsmath}
\usepackage{amssymb}
\usepackage{stmaryrd}
\usepackage[usenames,dvipsnames]{color}
\usepackage{enumerate,hyperref}

\ifdefined\LONGVERSION
  \relax
\else
\newcommand{\LONGVERSION}[1]{#1}
\newcommand{\SHORTVERSION}[1]{}
\fi

\LONGVERSION{
}

\newcommand{\LICSNOTbox}[1]{\LICS{}{\mbox{#1}}}

\newcommand{\LICSNR}{\LICS{\\[2ex]}{\qquad}}

\newcommand{\LICS}[2]{#2} 

\newdimen\proofrulebreadth \proofrulebreadth=.05em
\newdimen\proofdotseparation \proofdotseparation=1.25ex
\newdimen\proofrulebaseline \proofrulebaseline=2ex
\newcount\proofdotnumber \proofdotnumber=3
\let\then\relax
\def\hfi{\hskip0pt plus.0001fil}
\mathchardef\squigto="3A3B
\newif\ifinsideprooftree\insideprooftreefalse
\newif\ifonleftofproofrule\onleftofproofrulefalse
\newif\ifproofdots\proofdotsfalse
\newif\ifdoubleproof\doubleprooffalse
\let\wereinproofbit\relax
\newdimen\shortenproofleft
\newdimen\shortenproofright
\newdimen\proofbelowshift
\newbox\proofabove
\newbox\proofbelow
\newbox\proofrulename
\def\shiftproofbelow{\let\next\relax\afterassignment\setshiftproofbelow\dimen0 }
\def\shiftproofbelowneg{\def\next{\multiply\dimen0 by-1 }%
\afterassignment\setshiftproofbelow\dimen0 }
\def\setshiftproofbelow{\next\proofbelowshift=\dimen0 }
\def\setproofrulebreadth{\proofrulebreadth}

\def\prooftree{
\ifnum  \lastpenalty=1
\then   \unpenalty
\else   \onleftofproofrulefalse
\fi
\ifonleftofproofrule
\else   \ifinsideprooftree
        \then   \hskip.5em plus1fil
        \fi
\fi
\bgroup
\setbox\proofbelow=\hbox{}\setbox\proofrulename=\hbox{}%
\let\justifies\proofover\let\leadsto\proofoverdots\let\Justifies\proofoverdbl
\let\using\proofusing\let\[\prooftree
\ifinsideprooftree\let\]\endprooftree\fi
\proofdotsfalse\doubleprooffalse
\let\thickness\setproofrulebreadth
\let\shiftright\shiftproofbelow \let\shift\shiftproofbelow
\let\shiftleft\shiftproofbelowneg
\let\ifwasinsideprooftree\ifinsideprooftree
\insideprooftreetrue
\setbox\proofabove=\hbox\bgroup$\displaystyle 
\let\wereinproofbit\prooftree
\shortenproofleft=0pt \shortenproofright=0pt \proofbelowshift=0pt
\onleftofproofruletrue\penalty1
}

\def\eproofbit{
\ifx    \wereinproofbit\prooftree
\then   \ifcase \lastpenalty
        \then   \shortenproofright=0pt  
        \or     \unpenalty\hfil         
        \or     \unpenalty\unskip       
        \else   \shortenproofright=0pt  
        \fi
\fi
\global\dimen0=\shortenproofleft
\global\dimen1=\shortenproofright
\global\dimen2=\proofrulebreadth
\global\dimen3=\proofbelowshift
\global\dimen4=\proofdotseparation
\global\count255=\proofdotnumber
$\egroup  
\shortenproofleft=\dimen0
\shortenproofright=\dimen1
\proofrulebreadth=\dimen2
\proofbelowshift=\dimen3
\proofdotseparation=\dimen4
\proofdotnumber=\count255
}

\def\proofover{
\eproofbit 
\setbox\proofbelow=\hbox\bgroup 
\let\wereinproofbit\proofover
$\displaystyle
}%
\def\proofoverdbl{
\eproofbit 
\doubleprooftrue
\setbox\proofbelow=\hbox\bgroup 
\let\wereinproofbit\proofoverdbl
$\displaystyle
}%
\def\proofoverdots{
\eproofbit 
\proofdotstrue
\setbox\proofbelow=\hbox\bgroup 
\let\wereinproofbit\proofoverdots
$\displaystyle
}%
\def\proofusing{
\eproofbit 
\setbox\proofrulename=\hbox\bgroup 
\let\wereinproofbit\proofusing
\kern0.3em$
}

\def\endprooftree{
\eproofbit 
  \dimen5 =0pt
\dimen0=\wd\proofabove \advance\dimen0-\shortenproofleft
\advance\dimen0-\shortenproofright
\dimen1=.5\dimen0 \advance\dimen1-.5\wd\proofbelow
\dimen4=\dimen1
\advance\dimen1\proofbelowshift \advance\dimen4-\proofbelowshift
\ifdim  \dimen1<0pt
\then   \advance\shortenproofleft\dimen1
        \advance\dimen0-\dimen1
        \dimen1=0pt
        \ifdim  \shortenproofleft<0pt
        \then   \setbox\proofabove=\hbox{%
                        \kern-\shortenproofleft\unhbox\proofabove}%
                \shortenproofleft=0pt
        \fi
\fi
\ifdim  \dimen4<0pt
\then   \advance\shortenproofright\dimen4
        \advance\dimen0-\dimen4
        \dimen4=0pt
\fi
\ifdim  \shortenproofright<\wd\proofrulename
\then   \shortenproofright=\wd\proofrulename
\fi
\dimen2=\shortenproofleft \advance\dimen2 by\dimen1
\dimen3=\shortenproofright\advance\dimen3 by\dimen4
\ifproofdots
\then
        \dimen6=\shortenproofleft \advance\dimen6 .5\dimen0
        \setbox1=\vbox to\proofdotseparation{\vss\hbox{$\cdot$}\vss}%
        \setbox0=\hbox{%
                \advance\dimen6-.5\wd1
                \kern\dimen6
                $\vcenter to\proofdotnumber\proofdotseparation
                        {\leaders\box1\vfill}$%
                \unhbox\proofrulename}%
\else   \dimen6=\fontdimen22\the\textfont2 
        \dimen7=\dimen6
        \advance\dimen6by.5\proofrulebreadth
        \advance\dimen7by-.5\proofrulebreadth
        \setbox0=\hbox{%
                \kern\shortenproofleft
                \ifdoubleproof
                \then   \hbox to\dimen0{%
                        $\mathsurround0pt\mathord=\mkern-6mu%
                        \cleaders\hbox{$\mkern-2mu=\mkern-2mu$}\hfill
                        \mkern-6mu\mathord=$}%
                \else   \vrule height\dimen6 depth-\dimen7 width\dimen0
                \fi
                \unhbox\proofrulename}%
        \ht0=\dimen6 \dp0=-\dimen7
\fi
\let\doll\relax
\ifwasinsideprooftree
\then   \let\VBOX\vbox
\else   \ifmmode\else$\let\doll=$\fi
        \let\VBOX\vcenter
\fi
\VBOX   {\baselineskip\proofrulebaseline \lineskip.2ex
        \expandafter\lineskiplimit\ifproofdots0ex\else-0.6ex\fi
        \hbox   spread\dimen5   {\hfi\unhbox\proofabove\hfi}%
        \hbox{\box0}%
        \hbox   {\kern\dimen2 \box\proofbelow}}\doll%
\global\dimen2=\dimen2
\global\dimen3=\dimen3
\egroup 
\ifonleftofproofrule
\then   \shortenproofleft=\dimen2
\fi
\shortenproofright=\dimen3
\onleftofproofrulefalse
\ifinsideprooftree
\then   \hskip.5em plus 1fil \penalty2
\fi
}

\newcommand{\bla}{\ensuremath{\mbox{$$}}} 

\newcommand{\der}{\,\vdash}
\newcommand{\of}{\!:\!}

\newcommand{\FV}{\ensuremath{\mathsf{FV}}}

\newcommand{\Pot}{\mathcal{P}}
\newcommand{\dom}{\mathop{\mathsf{dom}}\nolimits}

\newcommand{\NN}{\mathbb{N}}

\newcommand{\beeq}{=_{\beta\eta}}

\def\lv{\mathopen{{[\kern-0.14em[}}}    
\def\rv{\mathclose{{]\kern-0.14em]}}}   

\newcommand{\dens}[1]{\mathopen{[\kern-0.3ex[}#1\mathclose{]\kern-0.3ex]}}
\newcommand{\denk}[2]{\mathopen{\{\kern-0.3ex|}#1\mathclose{|\kern-0.3ex\}}_{#2}}

\def\ltox#1{\buildrel\raise1pt\hbox{$\scriptstyle#1$}\over\longrightarrow}

\def\tocolow{\buildrel\raise-5pt\hbox{$\scriptscriptstyle+$}\over\rightarrow}

\def\lo{\mathopen{{\lceil\kern-0.25em\lceil}}}    
\def\ro{\mathclose{{\rfloor\kern-0.25em\rfloor}}}

\newcommand{\abbrev}[1]{#1} 

\newcommand{\eg}{\abbrev{e.\,g.}}

\newcommand{\ie}{\abbrev{i.\,e.}}

\newcommand{\wwlog}{w.\,l.\,o.\,g.} 
\newcommand{\Wlog}{W.\,l.\,o.\,g.}
\newcommand{\wrt}{w.\,r.\,t.}

\newcommand{\para}[1]{\paragraph*{\it#1}}
\newcommand{\paradot}[1]{\para{#1.}}

\newcommand{\mfor}{\ \mbox{for}\ }
\newcommand{\mforsome}{\ \mbox{for some}\ }

\newcommand{\miff}{\ \mbox{iff}\ }

\newcommand{\mand}{\ \mbox{and}\ }

\newcommand{\munless}{\ \mbox{unless}\ }

\newcommand{\mgoal}[1][]{\mbox{goal\ifthenelse{\equal{#1}{}}{}{~#1}}}

\newcommand{\rulename}[1]{\ensuremath{\mbox{\sc#1}}}
\newcommand{\ru}[2]{\dfrac{\begin{array}[b]{@{}c@{}} #1 \end{array}}{#2}}
\newcommand{\rux}[3]{\ru{#1}{#2}\ #3}
\newcommand{\nru}[3]{#1\ \ru{#2}{#3}}
\newcommand{\nrux}[4]{#1\ \ru{#2}{#3}\ #4}

\newcommand{\lcol}[1]{\multicolumn{1}{@{}l@{}}{{#1}}}

\newcommand{\subst}[3]{#3[#1/#2]}

\newcommand{\subscriptopt}[1]{\ifthenelse{\equal{#1}{}}{}{_{#1}}}
\newcommand{\superscriptopt}[1]{\ifthenelse{\equal{#1}{}}{}{^{#1}}}

\newcommand{\E}{\Erel{=}}

\newcommand{\squashtype}[1]{|\!|#1|\!|}
\newcommand{\sqval}[1]{[#1]}
\newcommand{\tsquash}{\sqval{\_}} 
\newcommand{\tsqelim}{\mathsf{sqelim}}
\newcommand{\subT}[3]{\{ #1 : #2 \mid #3 \}}
\newcommand{\Composite}{\mathsf{Composite}}

\newcommand{\close}[1]{\widehat{#1}}

\newcommand{\R}{\mathrel{\circledR}}
\newcommand{\Univ}[1]{\mathsf{U}_{#1}}

\newcommand{\epairT}[4][\evar]{(#2#1#3) \times #4}
\newcommand{\erpairT}{\epairT[\erased]}
\newcommand{\pairT}{\epairT[\!:\!]}
\newcommand{\letin}[4]{\tlet~(#1,#2) = #3 ~\tin~ #4}

\newcommand{\eappa}[1][\evar]{\,@^{#1}\,}

\newcommand{\appa}{\eappa[]}
 
\newcommand{\whnf}[1]{\mathord{\downarrow}#1}
\newcommand{\eqinfr}{\longleftrightarrow}
\newcommand{\eqchkr}{\Longleftrightarrow}
\newcommand{\eqinf}{\mathrel{\eqinfr}}
\newcommand{\eqchk}{\mathrel{\eqchkr}}
\newcommand{\eqty}{\eqchk}

\newcommand{\eqinftr}{\mathrel{\mbox{$\kern1.7ex\widehat{}\kern-2.6ex\longleftrightarrow$}}}
\newcommand{\eqchktr}{\mathrel{\mbox{$\kern1.7ex\widehat{}\kern-2.6ex\Longleftrightarrow$}}}
\newcommand{\eqinft}{\mathrel{\eqinftr}}
\newcommand{\eqchkt}{\mathrel{\eqchktr}}
\newcommand{\eqtyt}{\eqchkt}

\newcommand{\Axiom}{\mathsf{Axiom}}
\newcommand{\Rule}{\mathsf{Rule}}
\newcommand{\EPTS}{\textrm{EPTS}\xspace}

\newcommand{\tlet}{\mathsf{let}}
\newcommand{\tin}{\mathsf{in}}

\newcommand{\IITT}{\textrm{IITT}\xspace}

\newcommand{\ICCstar}{\textrm{ICC}\ensuremath{^*}\xspace}

\newcommand{\erof}{\div}

\newcommand{\evar}{\mathord{\star}}
\newcommand{\evarof}{\star}
\newcommand{\erased}{\mathord{\div}}
\newcommand{\noterased}{\mathord{:}}

\newcommand{\cext}[1]{#1.}
\newcommand{\erext}[3][\erased]{\cext{#2}\erann[#1]{#3}}
\newcommand{\eext}[2]{\erext[\evar]{#1}{#2}}
\newcommand{\erlam}[2]{\lambda [#1 \of #2].\,}
\newcommand{\erapp}[2]{#1[#2]}

\newcommand{\Set}[1][]{\tSet_{#1}}

\newcommand{\sid}{\mathsf{id}}

\newcommand{\evalsto}{\searrow}

\newcommand{\app}[1]{\mathsf{app}\,#1\,}

\newcommand{\funT}[2]{(#1 \of #2) \to}

\newcommand{\Nat}{\mathsf{Nat}}

\newcommand{\resurrect}[2][]{#2^{\ifthenelse{\equal{#1}{}}{\oplus}{
#1}}}
\newcommand{\eresurrect}[1]{\resurrect[\evar]{#1}}

\newcommand{\lam}[3]{\lambda #1 \of #2.\,#3}

\newcommand{\trefl}{\mathsf{refl}}

\newcommand{\tirr}{\mathord{*}}

\newcommand{\tSet}{\mathsf{Set}}
\newcommand{\Type}[1][]{\mathsf{Type}_{#1}}
\newcommand{\Prop}{\mathsf{Prop}}

\newcommand{\Sort}{\mathsf{Sort}}

\newcommand{\Exp}{\mathsf{Exp}}
\newcommand{\Ann}{\mathsf{Ann}}

\newcommand{\Cxt}{\mathsf{Cxt}}

\newcommand{\Wne}{\mathsf{Wne}}
\newcommand{\Whnf}{\mathsf{Whnf}}

\newcommand{\Var}{\mathsf{Var}}

\newcommand{\tfst}{\mathsf{fst}}
\newcommand{\tsnd}{\mathsf{snd}}

\newcommand{\eqpref}{eq-}

\newcommand{\reqeta}{\rulename{\eqpref{}$\eta$}}
\newcommand{\reqbeta}{\rulename{\eqpref{}$\beta$}}
\newcommand{\reqfun}{\rulename{\eqpref{}fun}}
\newcommand{\reqvar}{\rulename{\eqpref{}var}}
\newcommand{\reqlam}{\rulename{\eqpref{}lam}}
\newcommand{\reqapp}{\rulename{\eqpref{}app}}
\newcommand{\reqerapp}{\rulename{\eqpref{}[app]}}
\newcommand{\reqsort}{\rulename{\eqpref{}sort}}
\newcommand{\reqrefl}{\rulename{\eqpref{}refl}}
\newcommand{\reqsym}{\rulename{\eqpref{}sym}}
\newcommand{\reqtrans}{\rulename{\eqpref{}trans}}

\newcommand{\reqconv}{\rulename{\eqpref{}conv}}

\newcommand{\aleqpref}{aq-}
\newcommand{\raleqbeta}{\rulename{\aleqpref{}\kern-0.2ex$\beta$}}
\newcommand{\raleqbetal}{\rulename{\aleqpref{}\kern-0.2ex$\beta$-l}}
\newcommand{\raleqbetar}{\rulename{\aleqpref{}\kern-0.2ex$\beta$-r}}

\newcommand{\raleqapp}{\rulename{\aleqpref{}\kern-0.2ex app}}

\newcommand{\raleqforall}{\rulename{\aleqpref{}\kern-0.3ex$\forall$}}

\newcommand{\cempty}{()}

\newcommand{\A}{\mathcal{A}}
\newcommand{\B}{\mathcal{B}}

\newcommand{\DD}{\mathcal{D}}
\newcommand{\F}{\mathcal{F}}

\newcommand{\EL}{\mathcal{E}\kern-0.2ex\ell}

\newcommand{\PI}[1]{\Pi\,#1\,}
\newcommand{\PIAF}{\PI\A\F}

\newcommand{\x}{\mathsf{x}}
\newcommand{\xdel}[1][]{\ifthenelse{\equal{#1}{}}{\x_\Delta}{\x_{\Delta+#1}}}

\newcommand{\ONLYIFUP}[1]{}

\newcommand{\Seq}{\mathrel{\circledS}}
\newcommand{\Ceq}{\mathrel{\copyright}}

\newcommand{\closerel}[1]{\mathrel{\close{\small#1}}}
\newcommand{\cleq}{\closerel{\Seq}}

\newcommand{\Cleq}{\mathrel{\ensuremath{\raisebox{2pt}{\ensuremath{\close{\raisebox{-2pt}{\ensuremath{\Ceq}}}}}}}}

\newcommand{\circleEq}{\mathrel{\circledS}}

\newcommand{\T}{\mathcal{T}}

\newcommand{\Bool}{\mathsf{Bool}}
\newcommand{\ttrue}{\mathsf{true}}
\newcommand{\tfalse}{\mathsf{false}}

\newcommand{\VDash}{\mathrel{\mathord{|}\kern-0.15ex\mathord{\models}}}

\newcommand{\valic}{\Vdash^\mathsf{c}}
\newcommand{\valid}{\Vdash}

\newenvironment{deffigure}[2]{%
  \def\deffigurecaption{#2}%
  \begin{figure*}[htbp]%
  \begin{center}%
  \begin{minipage}{#1}%
  \hrule \vspace*{4ex}%
}{%
\vspace{2ex} \hrule%
\addvspace{2ex}%
  \end{minipage}%
  \end{center}%
  \caption{\deffigurecaption}%
  \end{figure*}%
}

\renewcommand{\para}[1]{\paragraph*{\it#1}#1}
\renewcommand{\paradot}[1]{\paragraph*{\it#1}}

\def\squareforqed{\ensuremath{\Box}}
\def\qed{\ifmmode\squareforqed\else{\unskip\nobreak\hfil
\penalty50\hskip1em\null\nobreak\hfil\squareforqed
\parfillskip=0pt\finalhyphendemerits=0\endgraf}\fi}

\renewcommand{\cempty}{\mathord{\diamond}}

\renewcommand{\funT}[2]{\Pi#1\of#2.\,}

\newcommand{\efunTsr}[6][\evar]{(#2#1#3) \stackrel{#5}{#4} #6}
\newcommand{\efunTss}[6][\evar]{\efunTsr[#1]{#2}{#3}{\to}{#4,#5}{#6}}
\newcommand{\efunTs}[5][\evar]{\efunTsr[#1]{#2}{#3}{\to}{#4}{#5}}
\newcommand{\efunT}[4][\evar]{\efunTs[#1]{#2}{#3}{}{#4}}

\newcommand{\erfunTs}{\efunTs[\erased]}
\newcommand{\erfunT}{\efunT[\erased]}

\newcommand{\funTs}{\efunTs[\!:\!]}
\renewcommand{\funT}{\efunT[\!:\!]}
\newcommand{\elam}[3][\evar]{\lambda #2#1#3.\,}
\renewcommand{\erlam}{\elam[\erased]}
\renewcommand{\lam}{\elam[\!:\!]}
\newcommand{\eapp}[1][\evar]{\,{}^{#1}}
\renewcommand{\erapp}{\eapp[\erased]}
\newcommand{\nerapp}{\eapp[:]}
\renewcommand{\app}{\eapp[]}

\renewcommand{\resurrect}[2][\div]{#2^{#1}}
\renewcommand{\eresurrect}[1]{\resurrect[\evar]{#1}}

\renewcommand{\eqinftr}{\mathrel{\mbox{$\kern2.1ex\widehat{}\kern-2.8ex\longleftrightarrow$}}}
\renewcommand{\eqchktr}{\mathrel{\mbox{$\kern2.1ex\widehat{}\kern-2.8ex\Longleftrightarrow$}}}

\newcommand{\cxtsep}{.\,}
\renewcommand{\eext}[4][\evar]{#2\cxtsep#3#1#4}
\renewcommand{\erext}{\eext[\!\erased\!]}
\renewcommand{\cext}{\eext[\!:\!]}

\renewcommand{\tirr}{\mathord{\bullet}}

\renewcommand{\E}{\mathrel{:=:}}

\newcommand{\hatDelta}{\Gamma}

\renewcommand{\valid}{{}\Vdash}

\renewcommand{\nru}[3]{\ru{#2}{#3}}
\renewcommand{\nrux}[4]{\ru{#2}{#3}\ #4}

\begin{document}
\title[Irrelevance in Type Theory]{
  On Irrelevance and Algorithmic Equality in \\
  Predicative Type Theory\rsuper*
}
\author[A.~Abel]{Andreas Abel\rsuper a}
\address{{\lsuper a}
  Department of Computer Science \\
  Ludwig-Maximilians-University Munich}
\email{andreas.abel@ifi.lmu.de} 

\author[G.~Scherer]{Gabriel Scherer\rsuper b}
\address{{\lsuper b}Gallium team, INRIA Paris-Rocquencourt}
\email{gabriel.scherer@gmail.com}

\titlecomment{{\lsuper*}Revision and extension of FoSSaCS 2011 conference publication.}

\begin{abstract}
\noindent
Dependently typed programs contain an excessive amount of static
terms which are necessary to please the type checker but irrelevant
for computation.   To separate static and dynamic code, several static
analyses and type systems have been put forward.   We consider
Pfenning's type theory with irrelevant quantification which is
compatible with a type-based notion of equality that respects
$\eta$-laws.  We extend Pfenning's theory to universes and large
eliminations and develop its meta-theory.  Subject reduction,
normalization and consistency are obtained by a Kripke model over the
typed equality judgement.  Finally, a type-directed equality algorithm
is described whose completeness is proven by a second Kripke model.
\end{abstract}

\maketitle

\section{Introduction and Related Work}
\label{sec:intro}

Dependently typed programming languages such as Agda
\cite{boveDybjerNorell:tphols09}, 
Coq \cite{inria:coq83}, and Epigram
\cite{mcBrideMcKinna:view}
allow the programmer to express
in one language programs, their types, rich invariants, and even
proofs of these invariants.  Besides code executed at run-time,
dependently typed programs contain 
much code
needed only to please 
the type checker, which is at the same time the
verifier of the proofs woven into the program.

Program extraction takes type-checked terms and discards parts that
are irrelevant for execution.  Augustsson's dependently typed
functional language Cayenne \cite{augustsson:cayenne} erases
\emph{types} using a universe-based analysis.
Coq's extraction procedure has been
designed by Paulin-Mohring and Werner \cite{paulinWerner:jsc93}
and Letouzey \cite{letouzey:types02} and discards not only
types but also proofs.
The erasure rests on Coq's
universe-based separation between propositional ($\Prop$)
and computational parts ($\Set/\Type$).  The rigid $\Prop/\Set$
distinction has the drawback of code duplication:  A structure
which is sometimes used statically and sometimes dynamically needs to
be coded twice, once in $\Prop$ and once in $\Set$.

An alternative to the fixed $\Prop/\Set$-distinction is
to let the usage context decide whether a term is a proof or a
program.  Besides whole-program analyses such as data flow, some type-based
analyses have been put forward.  One of them is Pfenning's modal type theory
of \emph{Intensionality, Extensionality, and Proof Irrelevance}
\cite{pfenning:intextirr}, later pursued by Reed \cite{reed:tphols03},
 which introduces functions with irrelevant
arguments that play the role of proofs.\footnote{
Awodey and Bauer \cite{awodeyBauer:propositionsAsTypes} give
a categorical treatment of proof irrelevance which is very similar to
Pfenning and Reed's.  However, they work in the setting of Extensional Type
Theory with undecidable type checking, we could not directly use their
results for this work.}
Not only can these arguments
be erased during extraction, they can also be disregarded in type
conversion tests during type checking.  This relieves the user of
unnecessary proof burden (proving that two proofs are equal).
Furthermore, proofs can not only be discarded during program
extraction but directly after type checking, since they will never be
looked at again during type checking subsequent definitions.

In principle, we have to distinguish ``post mortem'' program extraction,
let us call it \emph{external erasure}, and proof disposal during type
checking, let us call it \emph{internal erasure}.
External erasure deals with closed expressions, programs,
whereas internal erasure deals with open expressions that can have
free variables.  Such free variables might be assumed proofs of
(possibly false) equations and block type casts, or (possibly false)
proofs of well-foundedness and prevent recursive
functions from unfolding indefinitely.  For type checking to not go
wrong or loop, those proofs can only
be externally erased, thus, the $\Prop/\Set$ distinction is not for
internal erasure.  In Pfenning's type theory, proofs can never block
computations even in open expressions (other than computations on
proofs), thus, internal erasure is sound.

Miquel's Implicit Calculus of Constructions (ICC) \cite{miquel:tlca01} goes
further than Pfenning and considers also \emph{parametric} arguments as
irrelevant.  These are arguments which are irrelevant for function execution but
relevant during type conversion checking.  Such arguments may only be
erased in function application but not in the associated
type instantiation.  Barras and
Bernardo \cite{barrasBernardo:fossacs08} and Mishra-Linger and Sheard
\cite{mishraLingerSheard:fossacs08} have built decidable type systems
on top of ICC, but both have not fully integrated inductive types and
types defined by recursion (large eliminations).
Barras and Bernardo, as Miquel, have inductive types only in the form of
their impredicative encodings, Mishra-Linger \cite{mishraLinger:PhD}
gives introduction and elimination principles for inductive types by
example, but does not show normalization or consistency.

While Pfenning's type theory uses typed equality, ICC and its
successors interpret typed expressions as untyped $\lambda$-terms up
to untyped equality.  In our experience, the implicit quantification
of ICC, which allows irrelevant function arguments to appear
unrestricted in the codomain type of the function, is incompatible
with type-directed equality.  Examples are given in
Section~\ref{sec:examples}.  Therefore, we have chosen to scale
Pfenning's notion of proof irrelevance up to inductive types, and
integrated it into Agda.
  

In this article, we start with the ``extensionality and proof irrelevance''
fragment of Pfenning's type theory in Reed's version
\cite{reed:thesis,reed:tphols03}.
We extend it by a hierarchy of
predicative universes, yielding \emph{Irrelevant Intensional Type
  Theory} \IITT (Sec.~\ref{sec:syntax}).  After specifying a 
type-directed equality algorithm (Sec.~\ref{sec:algo}),
we construct a Kripke model 
for \IITT (Sec.~\ref{sec:sound}).  It allows us to prove  
normalization, subject reduction, and consistency, in one go
(Sec.~\ref{sec:meta}).
A second Kripke logical relation yields correctness of 
algorithmic equality and decidability of \IITT (Sec.~\ref{sec:compl}).  
Our models are ready for data types, large
eliminations, types with extensionality principles, and internal
erasure (Sec.~\ref{sec:ext}).

\subsection*{Contribution and Related Work}

We consider the design of our meta-theoretic argument as technical
novelty, although it heavily relies on previous works to which we owe our
inspiration.   Allen \cite{allen:PhD} describes a logical
relation for Martin-L\"of type theory with a countable universe hierarchy.
The seminal work of
Coquand \cite{coquand:conversion} describes an
untyped equality check for the Logical Framework and justifies it by a
logical relation for dependent types
that establishes subject reduction, normalization,
completeness of algorithmic equality, and injectivity of function
types in one go.  However, his approach cannot be easily extended to a
\emph{typed} algorithmic equality, due to problems with transitivity.

Goguen introduces \emph{Typed Operational Semantics} \cite{goguen:PhD}
to construct a \emph{Kripke} logical relation that 
simultaneously proves normalization, subject reduction, and confluence
for a variant of the \emph{Calculus of Inductive Constructions}.
From his results one can derive an equality check based on
reduction to normal form.  Goguen also shows how to derive syntactic
properties, such as closure of typing and equality under substitution,
by a Kripke-logical relation \cite{goguen:types00}.

Harper and Pfenning \cite{harperPfenning:equivalenceLF} popularize a
type-directed equality check for the Logical Framework
that scales to extensionality for unit types.  They prove completeness
of algorithmic equality by a Kripke model on \emph{simple types} which
are obtained by erasure from the dependent types.  Erasure is necessary
since algorithmic equality cannot be shown transitive before it is
proven sound; yet soundness hinges on subject reduction which rests on
function type injectivity which in turn is obtained from completeness
of algorithmic equality---a vicious cycle.  While erasure breaks the
cycle, it also prevents types to be defined by recursion on values
(so-called \emph{large eliminations}), a common feature of proof
assistants like Agda, Coq, and Epigram.

\emph{Normalization by evaluation} (NbE) has been successfully
used to obtain a type-directed equality check based on evaluation
in the context of dependent types with large eliminations
\cite{abelCoquandDybjer:lics07}.  
In previous work \cite{abelCoquandDybjer:mpc08}, the first author
applied NbE to justify a variant of Harper and Pfenning's algorithmic
equality \emph{without erasure}.  However, the meta-theoretic argument
is long-winded, and there is an essential gap in the proof of 
transitivity of the Kripke logical relation.

In this work, we explore a novel approach to justify type-directed
algorithmic equality for dependent types with predicative universes.
First, we show its soundness by a Kripke model built on top
of definitional equality.  The Kripke logical relation yields
normalization, subject reduction, and type constructor injectivity,
which also imply logical consistency of \IITT.
Further, it proves  syntactic properties
such as closure under substitution, following Goguen's lead
\cite{goguen:types00}.  The semantic proof of such syntactic
properties relieves us from
the deep lemma dependencies and abundant traps of
syntactic meta-theory of dependent types
\cite{harperPfenning:equivalenceLF,abelCoquand:fundinf07}. 
Soundness of algorithmic equality entails transitivity (which is
the stumbling stone), paving the way to show 
completeness of algorithmic equality by a second Kripke logical
relation, much in the spirit of Coquand \cite{coquand:conversion} and
Harper and Pfenning \cite{harperPfenning:equivalenceLF}.


This article is a revised and extended version of paper
\emph{Irrelevance in Type Theory with a Heterogeneous Equality
  Judgement} presented at the conference FoSSaCS~2011
\cite{abel:fossacs11}.  Unfortunately, the conference version has
inherited the above-mentioned gap \cite{abelCoquandDybjer:mpc08}
in the proof of transitivity of the
Kripke logical relation.  This is fixed in the present article by an
auxiliary Kripke model (Section~\ref{sec:sound}).  Further, we have dropped the
heterogeneous approach to equality in favor of a standard homogeneous
one.  Heterogeneous equality is not necessary for the style of
irrelevance we are embracing here.

\section{Irrelevant Intensional Type Theory}
\label{sec:syntax}

In this section, we present \emph{Irrelevant Intensional Type Theory}
\IITT which features two of Pfenning's function spaces
\cite{pfenning:intextirr}, the ordinary ``extensional'' $\funT x
U T$ and the proof irrelevant $\erfunT x U T$.  The main idea is that
the argument of a $\erfunT x U T$ function is counted as a proof and
can neither be returned nor eliminated on, it can only be passed as
argument to another proof irrelevant function or data constructor.
Technically, this is
realized by annotating variables as relevant, $x \of U$, or
irrelevant, $x \erof U$, in the typing context, 
to confine occurrences of
irrelevant variables to irrelevant arguments.


\paradot{Expression and context syntax}  We distinguish between
relevant ($t \nerapp u$ or simply $t \app u$) and irrelevant
application ($t \erapp u$).  Accordingly, we have relevant ($\lam x U
T$) and irrelevant ab\-strac\-tion ($\erlam x U T$).  Our choice of
typed abstraction is not fundamental; a bidirectional type-checking
algorithm \cite{coquand:type}
can reconstruct type and relevance annotations at
abstractions and applications.

\[
\begin{array}{l@{~}l@{~}l@{~}r@{~}l@{\quad}l}
  \Var   & \ni & x,y,X,Y \\ 
  \Sort  & \ni & s   & ::= & \Set[k] ~~(k \in \NN)
                  & \LICSNOTbox{universes}
   \\
  \Ann   & \ni & \evar & ::= & \erased \mid \noterased
                  & \LICSNOTbox{annotation: irrelevant, relevant} \\
  \Exp   & \ni & t,u,T,U
         & ::= & s \mid \efunTs{x}{U}{s,s'}{T}
                  & \LICSNOTbox{sort, (ir)relevant function type}
   \\ &&& \mid & x \mid \elam x U t \mid t \eapp u
                  & \LICSNOTbox{lambda-calculus}
   \\
  \Cxt   & \ni & \Gamma,\Delta & ::= & \cempty \mid \eext \Gamma x T
    & \LICSNOTbox{empty, (ir)relevant extension} \\
\end{array}
\]
Expressions are considered modulo $\alpha$-equality, we write $t
\equiv t'$ when we want to stress that $t$ and $t'$ identical (up to
$\alpha$). Similarly, we consider variables bound in a context
to be distinct, and when opening a term binder we will implicitly use
$\alpha$-conversion to add a fresh variable in the context.

For technical reasons, namely, to prove transitivity (Lemma~\ref{lem:translr})
of the Kripke logical relation  
in Section~\ref{sec:sound}, we
explicitly annotate function types $\efunTss x U s {s'} T$
with the sorts $s$ of domain $U$ and $s'$ of codomain $T$.
We may omit the annotation if it is inessential or determined by
the context of discourse.
In case $T$ does not mention $x$, we may write $U \to T$ for 
$\funT x U T$.

\paradot{Sorts}
\IITT is a pure type system (PTS) with infinite hierarchy of predicative
universes $\Set[0] : \Set[1] : ...$.  The universes are not cumulative.
We have the PTS axioms
$\Axiom = \{ (\Set[i],\Set[i+1]) \mid i \in \NN \}$ and the rules $\Rule =
\{(\Set[i],\Set[j],\Set[\max(i,j)]) \mid i,j \in \NN \}$.
As is customary, we will write
the side condition $(s,s') \in \Axiom$ just as $(s,s')$ and likewise
$(s_1,s_2,s_3) \in \Rule$ just as $(s_1,s_2,s_3)$.
$\IITT$ is a full and functional PTS,
which means that for all $s_1,s_2$ there is exactly one $s_3$ such that
$(s_1,s_2,s_3)$.  There is no subtyping, so that
types---and thus, sorts---are unique up to equality.  A proof of sort
unicity might relieve us from the sort annotation in function types, 
however, we obtain sort discrimination too late in our technical
development (Lemma~\ref{lem:tycondist}). 



\para{Substitutions} $\sigma$ are maps from variables to
expressions.  We require that the domain $\dom(\sigma) = \{ x \mid
\sigma(x) \not= x \}$ is finite.  We write $\sid$ for the identity
substitution and $\subst u x {}$ for the singleton substitution
$\sigma$ such that $\sigma(x) := u$ and $\sigma(y) := y$ for $y \neq x$.
Substitution extension $(\sigma, u/x)$ is formally defined as 
$\sigma \uplus \subst u x {}$.
Capture avoiding parallel substitution of $\sigma$ in $t$ is written as
juxtaposition $t \sigma$.

\para{Contexts} $\Gamma$ feature two kinds of bindings, relevant ($x
\of U$) and irrelevant ($x \erof U$) ones.  
The intuition, implemented by the typing rules below, 
is that only relevant variables are in scope in an expression.
\emph{Resurrection} $\resurrect\Gamma$ turns all irrelevant bindings
$(x \erof T)$ into the corresponding relevant ones $(x \of T)$
\cite{pfenning:intextirr}.
It is the tool to make irrelevant variables, also called proof
variables, available in proofs.
The generalization $\resurrect[\evar]\Gamma$ shall mean
$\resurrect \Gamma$ if $\evar = \mathord{\div}$, and just $\Gamma$
otherwise.
We write $\Gamma.\Delta$ for the concatenation of $\Gamma$ and
$\Delta$; herein, we suppose $\dom(\Gamma) \cap \dom(\Delta) = \emptyset$.

\paradot{Primitive judgements of \IITT}
The following three judgements are mutually inductively defined by the
rules given below and in Figure~\ref{fig:homall}.
\[
\begin{array}{l@{\qquad}l}
  \der \Gamma
    & \mbox{Context $\Gamma$ is well-formed.} \\
  \Gamma \der t : T
    & \mbox{\LICS{}{In context $\Gamma$,\ }expression $t$ has type $T$.} \\
  \Gamma \der t = t' : T
    & \mbox{In context $\Gamma$, $t$ and $t'$ are equal expressions
      of type $T$.} \\
\end{array}
\]


\paradot{Derived judgements}  To simplify notation, we introduce the
following four abbreviations:
\[
\begin{array}{lll}
  \Gamma \der t \div T
    & \miff & \resurrect\Gamma \der t : T ,\\
  \Gamma \der t = t' \div T
    & \miff &
      \Gamma \der t  \div T \mand
      \Gamma \der t' \div T ,\\
  \Gamma \der T
    & \miff & \Gamma \der T : s \mforsome s ,\\
  \Gamma \der T = T'
    & \miff & \Gamma \der T = T' : s \mforsome s
    .
    \\
\end{array}
\]
$\Gamma \der t \evarof T$ may mean $\Gamma \der t : T$ or 
$\Gamma \der t \div T$, depending on the value of placeholder $\evar$;
same for $\Gamma \der t = t' \evarof T$.
We sometimes write $\Gamma \der t,t' \evarof T$ to abbreviate the
conjunction of $\Gamma \der t \evarof T$ and $\Gamma \der t' \evarof
T$.  The notation $\Gamma \der T,T'$ is to be understood similarly.

\subsection{Rules}
Our rules for well-typed terms $\Gamma \der t : T$ extend 
Reed's rules \cite{reed:thesis} to PTS style.  There are only 6 rules;
we shall introduce them one-by-one.

\paradot{Variable rule} Only relevant variables can be extracted from
the context.
\[
  \ru{\der \Gamma \qquad
      (x \of U) \in \Gamma
    }{\Gamma \der x : U}
\] 
There is no variable rule for irrelevant bindings 
$(x \erof U) \in \Gamma$, in particular, the judgement 
$x \erof U \der x : U$ is not derivable.  
This essentially forbids proofs to appear in relevant positions.

\paradot{Abstraction rule} Relevant and irrelevant functions are
introduced analogously.
\[
  \ru{\eext \Gamma x U \der t : T
     \qquad \Gamma \der \efunTs x U {s,s'} T
    }{\Gamma \der \elam x U t : \efunTs x U {s,s'} T}
\]
To check a relevant function $\lam x U t$, we introduce a relevant
binding $x \of U$ into the context and continue checking the function
body $t$.  In case of an irrelevant function $\erlam x U t$, we
proceed with an irrelevant binding $x \erof U$.  This means that an
irrelevant function cannot computationally depend on its argument---it
is essentially a constant function.  In particular, 
$\erlam x U x$ is never well-typed.

As a side condition, we also need to check that the introduced
function type $\efunTs x U {s,s'} T$ is well-sorted; the rule is given below.

\paradot{Application rule}  
\[
  \ru{\Gamma \der t : \efunT x U T \qquad
      \Gamma \der u \evarof U
    }{\Gamma \der t \eapp u : \subst u x T}
\]

This rule uses our overloaded notations for bindings $\evarof$, that
can be specialized into two different instances for relevant and
irrelevant applications.

For relevant functions, we get the ordinary dependently-typed
application rule:
\[
  \ru{\Gamma \der t : \funT x U T \qquad
      \Gamma \der u : U 
    }{\Gamma \der t \app u : \subst u x T}
\]
When applying an irrelevant function, we resurrect the context before
checking the function argument. 
\[
  \ru{\Gamma \der t : \erfunT x U T \qquad
      \resurrect\Gamma \der u : U %
    }{\Gamma \der t \erapp u : \subst u x T}
\]
This means that irrelevant variables become relevant and can be used
in $u$.  The intuition is that the application $t\erapp u$ does not
computationally depend on $u$, thus, $u$ may refer to any variable,
even the ``forbidden ones''.  One may think of $u$ as a proof which
may refer to both ordinary and proof variables.

For example, let $\Gamma = f \of \erfunT y U U$.  Then the irrelevant
$\eta$-expansion $\erlam x U {f \erapp x}$ is well-typed in $\Gamma$,
with the following derivation:
\[
\begin{prooftree}
  \[
  \[
  \justifies
  \Gamma.~ x \erof U \der f : \erfunT y U U
  \]
  \[
  \justifies
  \Gamma.~ x : U \der x : U
  \]
  \justifies
  \Gamma.~ x \erof U \der f \erapp x : U
  \]
  \justifies
  \Gamma \der \lambda x \erof U.\, f \erapp x
     : \erfunT x U U
\end{prooftree}
\]
Observe how the status of $x$ changes for irrelevant to relevant when
we check the argument of $f$.

\paradot{Sorting rules}  These are the ``$\Axiom$s'' and the ``$\Rule$s'' of
PTSs to form types.
\[
   \ru{\der \Gamma
     }{\Gamma \der s : s'
     }{(s,s')}
\qquad
  \rux{\Gamma \der U : s_1 \qquad 
       \eext \Gamma x U \der T : s_2
     }{\Gamma \der \efunTs x U {s_1,s_2} T : s_3
     }{(s_1,s_2,s_3)}
\]
The rule for irrelevant function type formation follows Reed~\cite{reed:thesis}.
\[
  \rux{\Gamma \der U : s_1 \qquad  
       \erext \Gamma x U \der T : s_2
     }{\Gamma \der \erfunTs x U {s_1,s_2} T : s_3
     }{(s_1,s_2,s_3)}
\]
It states that the codomain of an irrelevant function cannot depend
relevantly on the function argument.  This fact is crucial for the
construction of our semantics in Section~\ref{sec:sound}.  Note that
it rules out \emph{polymorphism} in the sense of Barras and Bernado's
\emph{Implicit Calculus of Constructions}
\ICCstar \cite{barrasBernardo:fossacs08} and
Mishra-Linger and Sheard's \emph{Erasure Pure Type Systems} 
\EPTS \cite{mishraLingerSheard:fossacs08}; the type 
$\erfunT X {\Set[0]} {\funT x X X}$ is ill-formed in \IITT, but not in
\ICCstar or \EPTS.  In \EPTS, there is the following rule:
\[
  \rux{\Gamma \der U : s_1 \qquad
       \cext \Gamma x U \der T : s_2
     }{\Gamma \der \erfunTs x U {s_1,s_2} T : s_3
     }{(s_1,s_2,s_3)}
\]
It allows the codomain $T$ of an irrelevant function to arbitrarily depend
on the function argument $x$. This is fine in an erasure semantics, but
incompatible with our typed semantics in the presence of large
eliminations; we will detail the issues in examples \ref{ex:rel} and \ref{ex:largeElim}.

Another variant is Pfenning's rule for irrelevant function type formation
\cite{pfenning:intextirr}.
\[
  \rux{\Gamma \der U \div s_1 \qquad 
       \erext \Gamma x U \der T : s_2
     }{\Gamma \der \erfunTs x U {s_1,s_2} T : s_3
     }{(s_1,s_2,s_3)}
\] 
It allows the \emph{domain} of an irrelevant function to make use of
irrelevant variables in scope.  It does not give polymorphism, \eg, 
$\erfunT X {\Set[0]} {\funT x X X}$ is still ill-formed.  However, 
$\erfunT X {\Set[0]} {\erfunT x X X}$ would be well-formed.  
It is unclear how the equality rule for irrelevant function types
would look like---it is not given by Pfenning
\cite{pfenning:intextirr}.   The rule
\[
  \rux{\Gamma \der U = U' \div s_1 \qquad 
       \erext \Gamma x U \der T = T' : s_2
     }{\Gamma \der \erfunTs x U {s_1,s_2} T  
          = \erfunTs x {U'} {s_1,s_2} {T'}: s_3
     }{(s_1,s_2,s_3)}
\]
would mean that any two irrelevant function types are equal as long as
their codomains are equal---their domains are irrelevant.  This is not
compatible with our typed semantics and seems a bit problematic in
general.\footnote{This is why Reed \cite{reed:thesis} differs
  from Pfenning.}

\paradot{Type conversion rule} 
We have \emph{typed} conversion, thus, strictly speaking, \IITT is not
a PTS, but a \emph{Pure Type System with Judgemental Equality}
\cite{adams:jfp06}. 
\[
 \ru{\Gamma \der t : T \qquad
     \Gamma \der T = T'
   }{\Gamma' \der t : T'}
\]



\begin{deffigure}{\textwidth}{Rules of \IITT\label{fig:homall}}
\textbf{Context well-formedness.} \hfill \fbox{$\der \Gamma$} 
\begin{gather*}
  \ru{}{\der \cempty}
\qquad
  \ru{\der \Gamma \qquad
      \Gamma \der T 
    }{\der \eext \Gamma x T}
\end{gather*}
\textbf{Typing.} \hfill \fbox{$\Gamma \der t : T$}
\begin{gather*}
   \ru{\der \Gamma
     }{\Gamma \der s : s'
     }{(s,s')}
\qquad
  \rux{\Gamma \der U : s_1 \qquad  
       \eext \Gamma x U \der T : s_2
     }{\Gamma \der \efunTs x U {s_1,s_2} T : s_3
     }{(s_1,s_2,s_3)}
\\[2ex]
  \ru{\der \Gamma \qquad
      (x \of U) \in \Gamma
    }{\Gamma \der x : U}
\qquad
  \ru{\eext \Gamma x U \der t : T
     \qquad \Gamma \der \efunTs x U {s,s'} T
    }{\Gamma \der \elam x U t : \efunTs x U {s,s'} T}
\\[2ex]
  \ru{\Gamma \der t : \efunT x U T \qquad
      \Gamma \der u \evarof U
    }{\Gamma \der t \eapp u : \subst u x T}
\LICSNR
  \ru{\Gamma \der t : T \qquad
      \Gamma \der T = T'
    }{\Gamma' \der t : T'}
\end{gather*}
\textbf{Equality.} \hfill \fbox{$\Gamma \der t = t' : T$}  \\
Computation ($\beta$) and extensionality ($\eta$).
\begin{gather*}
 \nru{\reqbeta
    }{\eext \Gamma x U \der t : T  \qquad
      \Gamma \der u \evarof U 
    }{\Gamma \der (\elam x U t) \eapp u = \subst u x t : \subst u x T}
\qquad
 \nru{\reqeta
    }{\Gamma \der t : \efunTs x U {s,s'} T 
    }{\Gamma \der t = \elam x U {t \eapp x} : \efunTs x U {s,s'} T
    }
\end{gather*}
Equivalence rules.
\begin{gather*}
  \nru{\reqrefl
     }{\Gamma \der t : T
     }{\Gamma \der t = t : T}
\qquad
  \nru{\reqsym
     }{\Gamma \der t = t' : T
     }{\Gamma \der t' = t : T}
\qquad
  \nru{\reqtrans
     }{\Gamma \der t_1 = t_2 : T \qquad
       \Gamma \der t_2 = t_3 : T
     }{\Gamma \der t_1 = t_3 : T}
\end{gather*}
Compatibility rules.
\begin{gather*}
 \nrux{\reqfun
     }{\Gamma \der U = U' : s_1 \qquad
       \eext \Gamma x U \der T = T' : s_2
     }{\Gamma \der \efunTs x U {s_1,s_2} T = \efunTs x {U'}{s_1,s_2}{T'} : s_3
     }{(s_1,s_2,s_3)}
\\[2ex]
  \nru{\reqlam
     }{\Gamma \der U = U' : s_1
       \qquad \eext \Gamma x U \der T : s_2
       \qquad \eext \Gamma x U \der t = t' : T
     }{\Gamma \der \elam x U t = \elam x {U'} {t'} : \efunTs x U {s_1,s_2} T
     }
\\[2ex]
  \nru{\reqapp
     }{\Gamma \vdash t = t' \of \efunT x U T \qquad
       \Gamma \der u = u ' \evarof U
     }{\Gamma \der t \eapp u = t' \eapp u' : \subst u x T
     }
\end{gather*}
Conversion rule.
\begin{gather*}
  \nru{\reqconv
      }{\Gamma \der t = t' : T \qquad
        \Gamma \der T = T' 
      }{\Gamma \der t = t' : T'}
\end{gather*}
\end{deffigure}


\paradot{Equality}
Figure~\ref{fig:homall} recapitulates the typing rules and lists the
rules to derive context well-formedness $\der \Gamma$ and equality
$\Gamma \der t = t' : T$.  Equality is the least congruence over the
$\beta$- and $\eta$-axioms.  Since equality is typed we can extend
$\IITT$ to include an extensional unit type (Section~\ref{sec:ext}).
Let us inspect the congruence rule for application:
\[
  \nru{\reqapp
     }{\Gamma \vdash t = t' \of \efunT x U T \qquad
       \Gamma \der u = u ' \evarof U
     }{\Gamma \der t \eapp u = t' \eapp u' : \subst u x T
     }
\]
In case of relevant functions ($\evar = \noterased$) we obtain the
usual dependently-typed application rule of equality.  Otherwise, we
get:
\[
  \nru{\reqerapp
     }{\Gamma  \vdash t = t' \of \erfunT x U T \qquad
       \resurrect\Gamma \der u : U \qquad
       \resurrect{\Gamma} \der u' : U
     }{\Gamma \der t \erapp u = t' \erapp u' : \subst u x T
     }
\]
Note that the arguments $u$ and $u'$ to the irrelevant functions need to be
well-typed but not related to each other.  This makes precise the
intuition that $t$ and $t'$ are constant functions.



\subsection{Simple properties of \IITT}

In the following, we prove two basic invariants of derivable 
\IITT-judgements:  The context is always well-formed, and judgements
remain derivable under well-formed context extensions (weakening).
\begin{lem}[Context well-formedness]\label{lem:wf} \bla
  \begin{enumerate}[\em(1)]
  \item If $\der \cext \Gamma x U \cxtsep \Gamma'$ then $\Gamma \der U$.
  \item If $\Gamma \der t : T$ or $\Gamma \der t = t' : T$ 
    then $\der\Gamma$.
  \end{enumerate}
\end{lem}
\begin{proof}
  By a simple induction on the derivations.
\end{proof}

It should be noted that we only prove the most basic well-formedness
statements here.  One would expect that $\Gamma \der t : T$ or $\Gamma
\der t = t' : T$ also implies $\Gamma \der T$, or that $\Gamma \der t
= t' : T$ implies $\Gamma \der t : T$. This is true---and we will refer
to these implications as \emph{syntactic validity}---but this cannot be
proven without treatment of substitution, due to the typing rule for
application, which requires substitution in the type, and due to the
equality rule for a $\beta$-redex, which uses substitution in both
term and type. Therefore, syntactic validity is delayed until
Section~\ref{sec:sound} 
(Corollary~\ref{cor:synval}), where substitution will be handled by
semantic, rather than syntactic, methods.

\paradot{Weakening} 
We can weaken a context $\Gamma$ by adding bindings or making
irrelevant bindings relevant. Formally, we have an order on binding
annotations, which is the order induced by $\noterased \leq
\erased$, and we define weakening by monotonic extension.

A well-formed context $\der \Delta$ \emph{extends} a well-formed context
$\der \Gamma$, written $\Delta \leq \Gamma$, if and only if:
\[ \forall x \in \dom(\Gamma),\quad
   (x \evarof_1 U) \in \Gamma \implies
     (x \evarof_2  U) \in \Delta \text{ with } \evarof_1 \leq \evarof_2
     .
\]
Note that this allows to insert new bindings or relax existing ones
at any position in $\Gamma$, not just at the end.

\begin{lem}[Weakening] \label{lem:weak} \bla
  Let $\Delta \leq \Gamma$.
  \begin{enumerate}[\em(1)]
  \item \label{it:weakcxt}
    If $\der \Gamma.\Gamma'$ and $\dom(\Delta) \cap \dom(\Gamma')
    = \emptyset$ then $\der \Delta.\Gamma'$.
  \item If $\Gamma \der t : T$ then $\Delta \der t : T$.
  \item If $\Gamma \der t = t' : T'$ then $\Delta \der t = t' : T$.
  \end{enumerate}
\end{lem}
\begin{proof}
Simultaneously by induction on the derivation.
Let us look at some cases:
\begin{desCription}

\item\noindent{\hskip-12 pt\em Case\enspace}
\[
 \nrux{\reqsort
     }{\der \Gamma
     }{\Gamma \der s = s : s'
     }{(s,s')}
\]
By assumption $\der \Delta$, thus $\Delta \der s = s : s'$.

\item\noindent{\hskip-12 pt\em Case\enspace}
\[
 \nru{\reqvar
    }{(x \of U) \in \Gamma
      \quad
      \der \Gamma 
   }{\Gamma \der x = x : U}
\]
Since $\Delta \leq \Gamma$ we have $(x \of U) \in \Delta$, thus
$\Delta \der x = x : U$.

\item\noindent{\hskip-12 pt\em Case\enspace}
\[
  \nru{\reqlam
     }{\Gamma \der U = U' : s_1
       \qquad \eext \Gamma x U \der T : s_2
       \qquad \eext \Gamma x U \der t = t' : T
     }{\Gamma \der \elam x U t = \elam x {U'} {t'} : \efunTs x U {s_1,s_2} T
     }
\]
\Wlog, $x \not\in \dom(\Delta)$. 
By (\ref{it:weakcxt}) and definition of context weakening,
$\Delta \leq \Gamma$ implies $\eext \Delta x U \leq \eext \Gamma x U$,
so all premises can be appropriately weakened by induction hypothesis.
\end{desCription}
\end{proof}

\subsection{Examples}
\label{sec:examples}



\begin{exa}[Relevance of types]\footnote{Example suggested by a
    reviewer of this paper.} \label{ex:rel}
  We can extend \IITT by a unit type $1$ with extensionality
  principle.
  \begin{gather*}
    \ru{\der \Gamma
      }{\Gamma \der 1 : \Set[i]}
\qquad
    \ru{\der \Gamma
      }{\Gamma \der () : 1}
\qquad
    \ru{\Gamma \der t : 1 \qquad \Gamma \der t' : 1
      }{\Gamma \der t = t' : 1}
  \end{gather*}
  Typed equality allows us to equate all inhabitants of the unit
  type.  As a consequence, the Church numerals over the unit type all
  coincide, \eg,
\[
\begin{array}{r@{~}l}
  \Gamma \der &
  \lam f {1 \to 1} \lam x 1 x \\ = &
  \lam f {1 \to 1} \lam x 1 f \app x ~:~ (1 \to 1) \to 1 \to 1
  .
\end{array}
\]
  In systems with untyped equality, like \ICCstar and \EPTS, these
  terms erase to untyped Church-numerals $\lambda f \lambda x.x$ and
  $\lambda f \lambda x.\, f\app x$ and are necessarily distinguished.

  If we trade the unit type for $\Bool$ or any other type with more
  than one inhabitant, the two terms become different in \IITT.  
  This means that in \IITT, types are relevant, and we need to reject
  irrelevant quantification over types like in 
  $\erfunT X {\Set[0]} {(X \to X) \to X \to X}$.
  In \IITT, the polymorphic types of Church numerals are
  $\funT X {\Set[i]} {(X \to X) \to X \to X}$.
\end{exa}
\begin{exa}[$\Sigma$-types]
  \IITT can be readily extended by weak $\Sigma$-types.
  \begin{gather*}
    \rux{\Gamma \der U : s_1 \qquad
         \eext \Gamma x U \der T : s_2
       }{\Gamma \der \epairT x U T : s_3
       }{(s_1,s_2,s_3)}
\\[0.5em]
   \ru{\Gamma \der u \evarof U \qquad
       \Gamma \der t : \subst u x T \qquad
       \Gamma \der \epairT x U T
     }{\Gamma \der (u,t) : \epairT x U T} 
\\[0.5em]
  \ru{\Gamma \der p : \epairT x U T \qquad
      \cext {\eext \Gamma x U} y T \der v : V
    }{\Gamma \der \letin x y p v : V}
\\[0.5em]
  \ru{\Gamma \der u \evarof U \qquad
      \Gamma \der t : \subst u x T \qquad
      \cext {\eext \Gamma x U} y T \der v : V \qquad
       \Gamma \der \epairT x U T
    }{\Gamma \der (\letin x y {(u,t)} v) = \subst t y {\subst u x v} : V}
  \end{gather*}
Additional laws for equality could be considered, like commuting
conversions, or the identity $(\letin x y p {(x,y)}) = p$.  The
relevant form $\pairT x U T$ admits a strong version with projections
$\tfst$ and $\tsnd$ and full extensionality 
$p = (\tfst\,p,\, \tsnd\,p) : \pairT x U T$.  However, strong
irrelevant $\Sigma$-types $\erpairT x U T$ are problematic because of
the first projection:
\[
  \ru{\Gamma \der p : \erpairT x U T
    }{\Gamma \der \tfst\,p \erof U} 
\]
With our definition of $\Gamma \der u \erof U$ as $\resurrect\Gamma
\der u : U$, this rule is misbehaved: it allows us get hold of an
irrelevant value in a relevant context.  We could define a closed 
function $\pi_1 : \erpairT x U 1 \to U$, and composing it with
$(\_\,,()) : \erfunT x U {\erpairT x U 1}$ would give us an identity
function of type $\erfunT x U U$ which magically makes irrelevant
things relevant and \IITT inconsistent.
In this article, we will not further consider strong $\Sigma$-types
with irrelevant components; we leave the in-depth investigation to
future work.
\end{exa}

\begin{exa}[Squash type]
  The \emph{squash} type $\squashtype T$ was first introduced 
  in the context of NuPRL
  \cite{constable:nuprl}; it contains exactly one inhabitant iff $T$
  is inhabited.  Semantically, one obtains $\squashtype T$ from $T$ by
  equating all of $T$'s inhabitants.  In \IITT, we can define
  $\squashtype T$ as internalization of the irrelevance modality, as
  already suggested by Pfenning \cite{pfenning:intextirr}.  The first
  alternative is via the weak irrelevant $\Sigma$-type.
\[
\begin{array}{lll}
  \squashtype \_ & :  & \Set[i] \to \Set[i] \\
  \squashtype T  & := & \erpairT \_ T 1 \\[0.5em]
  \tsquash      & :  & \erfunT x T {\squashtype T} \\
  \sqval x      & := & (x,()) \\[0.5em]
  \multicolumn 3 l {\tsqelim~ (T : \Set[i])~
           (P : \squashtype T \to \Set[j])~
           (f : \erfunT x T {P \, \sqval x })~
           (t : \squashtype T)}\\
    & :  & P \app t \\
    & := & \letin x \_ t {f \erapp x} \\
\end{array}
\]
It is not hard to see that $\squashtype \_$ is a monad.  
All \emph{canonical} inhabitants of $\squashtype T$ are
definitionally equal:
\[
\ru{\Gamma \der t, t' \div T
  }{\Gamma \der \sqval t = \sqval{t'} : 
       \squashtype T}
\]
This is easily shown by expanding the definition of $\tsquash$ and
using the congruence rule for pairs with an irrelevant first
component.

However, we cannot show that \emph{all} inhabitants of $\squashtype T$
are definitionally equal, because of the missing extensionality
principles for weak $\Sigma$.  Thus, the second alternative is to add the
squash type to \IITT via the rules:
\begin{gather*}
  \ru{\Gamma \der T : \Set[i]
    }{\Gamma \der \squashtype T : \Set[i]}
\qquad
  \ru{\Gamma \der t \div T
    }{\Gamma \der \sqval t : \squashtype T}
\qquad
  \ru{\Gamma \der t : \squashtype T \qquad
      \erext \Gamma x T \der v : V
    }{\Gamma \der \tlet~ \sqval x = t ~\tin~ v : V} 
\\[0.5em]
  \ru{\Gamma \der t, t' : \squashtype T 
    }{\Gamma \der t = t' : \squashtype T}
\qquad
  \ru{\Gamma \der t \div T \qquad
      \erext \Gamma x T \der v : V        
    }{\Gamma \der (\tlet~\sqval x = \sqval t ~\tin~ v) = \subst t x v : V}
\end{gather*}
Our model (Section~\ref{sec:sound}) is ready to interpret these rules,
as well as normalization-by-evaluation inspired models 
\cite{abelCoquandPagano:lmcs11}.  
\end{exa}

\begin{exa}[Subset type]
  The subset type $\subT x U T$ is definable from $\Sigma$ and squash as 
  $\pairT x U {\squashtype T}$.
\end{exa}
To discuss the next example, we consider a further extension of \IITT by
Leibniz equality  and natural numbers:
\[
\begin{array}{l@{~}l@{~}l@{\quad}l}
  a \equiv b & : & \Set[i]
    & \mfor A : \Set[i] \mand a,b : A \\
  \trefl & : & a \equiv a
    & \mfor A : \Set[i] \mand a : A \\[0.5em]
  \Nat & : & \Set[i] \\
  0,1,\dots & : & \Nat \\
  +,* & : & \multicolumn 2 l {\Nat \to \Nat \to \Nat .} \\
\end{array}
\]
\begin{exa}[Composite]\footnote{Example suggested by reviewer.}
Let the set of composite numbers 
$\{4,6,8,9,10,12,14,15,\dots\}$ be numbers that are the product of two
natural numbers $\geq 2$.
\[
  \Composite = \subT n \Nat { \pairT k \Nat \pairT l \Nat 
    (n \equiv (k+2)*(l+2)) }
\]  
Most composite numbers have several factorizations, and thanks to
irrelevance the specific composition is ignored when handling
composite numbers.  For instance, 12 as product of 3 and 4 is not
distinguished from the 12 as product of 2 and 6.
\[
  (12, \sqval{(1,(2,\trefl))}) = (12, \sqval{(0,(4,\trefl))})
  : \Composite
  .
\]
\end{exa}

\begin{exa}[Large eliminations]\label{ex:largeElim}%
  \footnote{Inspired by discussions with Ulf Norell during the 11th Agda
    Implementers' Meeting.}
\renewcommand{\T}{\mathsf{T}}
  The $\ICCstar$ \cite{barrasBernardo:fossacs08} or
  $\EPTS$ \cite{mishraLingerSheard:fossacs08} irrelevant function
  type $(x \erof A) \to B$ allows $x$ to appear \emph{relevantly} in
  $B$.  This extra power raises some issues with large eliminations.
  Consider
\[
\begin{array}{l@{~}l@{~}l}
  \T         & : & \Bool \to \Set[0] \\
  \T~\ttrue  & = & \Bool \to \Bool \\
  \T~\tfalse & = & \Bool \\[0.5em]
  t          & = &
    \lambda F : \erfunT b \Bool (\T\,b \to \T\,b) \to \Set[0]. \\ &&
    \lambda g : (F \erapp \tfalse~(\lambda x : \Bool.\, x)) \to \Bool. \\ &&
    \lambda a : F \erapp \ttrue~(\lambda x : \Bool \to \Bool. \lambda y :
    \Bool.\, x\,y).~ g\,a
    .
\end{array}
\]
  The term $t$ is well-typed in $\ICCstar+\T$ because the domain type
  of $g$ and the type of $a$ are $\beta\eta$-equal after erasure
  $(-)^*$ of
  type annotations and irrelevant arguments:
\begin{align*}
  (F \erapp \tfalse~(\lambda x : \Bool.\, x))^* & = F~(\lambda x x) \\ \beeq
   F~(\lambda x \lambda y.\,x\,y)  & =
  (F \erapp \ttrue~(\lambda x : \Bool \to \Bool. \lambda y : \Bool.\, x\,y))^*
\end{align*}
  While a Curry view supports this, it is questionable
  whether identity functions at different types should be viewed as
  one.  It is unclear how a type-directed equality algorithm (see
  Sec.~\ref{sec:algo}) should proceed here; it needs to recognize that
  $x : \Bool$ is equal to $\lambda y \of \Bool.\,x\,y : \Bool \to
  \Bool$.
  This situation is amplified by a unit type $1$ with extensional
  equality.  When we change $\T\,\ttrue$ to $1$ and the type of $a$ to
  $F \erapp \ttrue~(\lambda x \of 1.\, ())$ then $t$ should still type-check,
  because $\lambda x.\,()$ is the identity function on $1$.  However,
  $\eta$-equality for $1$ cannot be checked without types, and a
  type-directed algorithm would end up checking (successfully)
  $x : \Bool$ for equality with $() : 1$.  
  This algorithmic equality cannot be transitive, because then any two
  booleans would be equal.

  Summarizing, we may conclude that the type of $F$ bears trouble
  and needs to be rejected.  \IITT does this because it forbids the
  irrelevant $b$ in relevant positions such as $\T\,b$; \ICCstar
  lacks $\T$ altogether.
  Extensions of $\ICCstar$ should at least make sure that $b$ is
  never eliminated, such as in $\T\,b$.  Technically, $\T$ would have
  to be put in a separate class of \emph{recursive} functions, those
  that actually compute with their argument.  We leave the interaction
  of the three different function types to future research.
\end{exa}

\section{Algorithmic Equality}
\label{sec:algo}

The algorithm for checking equality in \IITT is inspired by Harper and
Pfenning \cite{harperPfenning:equivalenceLF}.  Like theirs, it is
type-directed, but we are using the full dependent type and not an
erasure to simple types (which would anyway not work due to large
eliminations).  We give the algorithm in form of judgements and rules
in direct correspondence to a functional program.  


Algorithmic equality is meant to be used as part of a type checking
algorithm. It is the algorithmic counterpart of the definitional
conversion rule; in particular, it will only be called on terms that
are already know to be well-typed -- in fact, types that are
well-sorted. We rely on this precondition in the algorithmic
formulation.

Algorithmic equality consists of three interleaved
judgements. A \emph{type equality} test checks equality between two
types, by inspecting their weak head normal forms. Terms found inside
dependent types are reduced and the resulting neutral terms are
compared by \emph{structural equality}. The head variable of such
neutrals provides type information that is then used to check the
(non-normal) arguments using \emph{type-directed equality}, by
reasoning on the (normalized) type structure to perform
$\eta$-expansions on product types. After enough expansions, a base
type is reached, where structural equality is called again, or a sort,
at which we use type equality.

Informally, the interleaved reductions are the algorithmic
counterparts of the $\beta$-equality axiom, the type and structural
equalities account for the compatibility rules, and type-directed
equality corresponds to the $\eta$-equality axiom. The remaining
equivalence rules are emergent global properties of the algorithm.

\paradot{Weak head reduction}

Weak head normal forms (whnfs) are given by the following grammar:
\[
\begin{array}{lllll@{\quad}l}
  \Whnf & \ni & a,b,f,A,B,F & ::= & s \mid \efunTs{x}{U}{s,s'}{T}
    \mid \elam x U t \mid n
    & \mbox{whnf} \\
  \Wne & \ni & n,N & :: = & x \mid n \eapp u
    & \mbox{neutral whnf} \\
\end{array}
\]
Weak head evaluation
$t \evalsto a$ and active application
$f \eappa  u \evalsto a$ are functional relations
given by the following rules.
\begin{gather*}
  \ru{t \evalsto f \qquad
      f \eappa u \evalsto a
    }{t \eapp u \evalsto a}
\qquad
  \ru{}{a \evalsto a}
\qquad
  \ru{\subst u x t \evalsto a
    }{(\elam x U t) \eappa u \evalsto a}
\qquad
  \ru{
    }{n \eappa u \evalsto n \eapp u}
\end{gather*}
Instead of writing the propositions $t \evalsto a$ and $P[a]$
we will sometimes simply write $P[\whnf t]$.  Similarly, we might
write $P[f\eappa u]$ instead of $f \eappa u \evalsto a$ and $P[a]$.
In rules, it is understood that the evaluation judgement is always an
extra premise, never an extra conclusion.

Algorithmic equality is given as type equality, structural equality,
and type-directed equality, which are mutually recursive.  The
equality algorithm is only invoked on well-formed expressions of the
correct type.

\para{Type equality}
$\Delta \der A \eqty A'$, for weak head normal forms, and
$\Delta \der T \eqtyt T'$, for arbitrary well-formed types,
checks that two given types are equal in their respective contexts.
\begin{gather*}
  \ru{\Delta  \der \whnf T \eqty \whnf{T'}
    }{\Delta  \der T \eqtyt T'}
\qquad
  \ru{\Delta \der N \eqinft N' : T
    }{\Delta \der N \eqty N'}
\\[0.5em]
  \ru{}{\Delta \der s \eqty s}
\qquad
  \ru{\Delta  \der U \eqtyt U' \qquad
      \cext \Delta x U \der T \eqtyt T'
    }{\Delta  \der \efunTs x U {s,s'} T \eqty \efunTs x {U'}{s,s'}{T'}
    }
\end{gather*}
Note that when invoking structural equality on neutral types $N$ and
$N'$, we do not care which type $T$ is returned, since we know by
well-formedness that $N$ and $N'$ must have the same sort.

\para{Structural equality}
$\Delta \der n \eqinf n' : A$ and
$\Delta \der n \eqinft n' : T$ checks the neutral
expressions $n$ and $n'$ for equality and at the same time infers
their type, which is returned as output.
\begin{gather*}
  \ru{\Delta \der n \eqinft n' : T
    }{\Delta \der n \eqinf  n' : \whnf T
    }
\qquad
  \ru{(x \of T) \in \Delta 
    }{\Delta \der x  \eqinft x : T}
\\[0.5em]
  \ru{\Delta  \der n \eqinf n'  : \funT x U T \qquad
      \Delta  \der u \eqchkt u' : U
    }{\Delta  \der n  \app u \eqinft n' \app u' : \subst u x T
    }
\qquad
  \ru{\Delta \der n \eqinf n' : \erfunT x U T
    }{\Delta \der n \erapp u \eqinft n' \erapp {u'} : \subst u x T
    }
\end{gather*}

\para{Type-directed equality}
$\Delta \der t \eqchk t' : A$ and
$\Delta \der t \eqchkt t' : T$ checks terms $t$ and
$t'$ for equality and proceeds by the structure of the supplied
type, to account for $\eta$.
\begin{gather*}
  \ru{\Delta \der t \eqchk  t' : \whnf T
    }{\Delta \der t \eqchkt t' : T}
\qquad
  \ru{\eext \Delta x {U} \der t \eapp x \eqchkt t' \eapp x : T
    }{\Delta \der t \eqchk t' : \efunT x U T
    }
\\[0.5em]
  \ru{\Delta \der T \eqtyt T'
    }{\Delta \der T \eqchk T' : s}
\qquad 
  \ru{\Delta \der \whnf t \eqinft \whnf{t'} : T
    }{\Delta \der t \eqchk t' : N}
\end{gather*}
Note that in the but-last rule we do not check that the inferred
type $T$ of $\whnf t$ equals the ascribed type $N$.  Since algorithmic
equality is only invoked for well-typed $t$,
we know that this must always be the case.  Skipping
this test is a conceptually important improvement over Harper and
Pfenning~\cite{harperPfenning:equivalenceLF}.

Due to dependent typing, it is not obvious that algorithmic equality
is symmetric and transitive.  For instance, consider symmetry in case
of application:  We have to show that 
$\Delta \der n' \app u' \eqinft n \app u : \subst{u}x{T}$, 
but using the induction hypothesis we obtain this equality only at
type $\subst {u'} x T$.  To conclude, we need to convert types, which is
only valid if we know that $u$ and $u'$ are actually equal.  Thus, we
need soundness of algorithmic equality to show its transitivity.
Soundness \wrt\ declarative equality requires subject reduction, which
is not trivial, due to its dependency on function type injectivity.  
In the next section (\ref{sec:sound}), we construct by a Kripke
logical relation which gives us subject reduction and soundness of
algorithmic equality (Section~\ref{sec:meta}), and, finally, symmetry
and transitivity of algorithmic equality.
 
A simple fact about algorithmic equality is that the inferred types
are unique up to syntactic equality (where we consider
$\alpha$-convertible expressions as identical).
Also, they only depend on the left hand side neutral term $n$.
\begin{lem}[Uniqueness of inferred types]\label{lem:uniqinfer} \bla
  \begin{enumerate}[\em(1)]
  \item
  If  $\Delta \der n \eqinf n_1 : A_1$
  and $\Delta \der n \eqinf n_2 : A_2$
  then $A_1 \equiv A_2$.
  \item
  If  $\Delta \der n \eqinft n_1 : T_1$
  and $\Delta \der n \eqinft n_2 : T_2$
  then $T_1 \equiv T_2$.
  \end{enumerate}
\end{lem}
\noindent
Extending structural equality to irrelevance, we let
\[
\ru{\resurrect\Delta \der n \eqinf n : A \qquad
    \resurrect{\Delta} \der n' \eqinf n' : A
  }{\Delta \der n  \eqinf n' \div A}
\]
and analogously for
$\Delta \der n  \eqinft n' \div T$.

\section{A Kripke Logical Relation for Soundness}
\label{sec:sound}

In this section, we construct a Kripke logical relation in the spirit
of Goguen \cite{goguen:types00} and Vanderwaart and Crary
\cite{vanderwaartCrary:lfm02} that proves weak head normalization,
function type injectivity, and subject reduction plus syntactical
properties like substitution in judgements and syntactical validity.
As an important consequence, we obtain soundness of algorithmic
equality \wrt\ definitional equality.  This allows us to establish
that algorithmic equality on well-typed terms is a partial equivalence
relation.

\subsection{An Induction Measure}

Following Goguen \cite{goguen:PhD} and previous work
\cite{abelCoquandDybjer:mpc08}, we first define a semantic universe
hierarchy $\Univ i$ whose sole purpose is to provide a measure for
defining a logical relation and proving some of its properties.  The
limit $\Univ \omega$ corresponds to the proof-theoretic strength or
ordinal of \IITT.

We denote sets of expressions by $\A,\B$ and functions from
expressions to sets of expressions by $\F$.
Let $\close\A = \{ t \mid \whnf t \in \A \}$
denote the closure of $\A$ by weak head expansion.
The dependent function space is defined as
$\PIAF = \{ f \in \Whnf \mid
\forall u \in \close\A.\, f \appa u \in \F(u) \}$.

By recursion on $i \in \NN$ we define inductively sets $\Univ i \subseteq \Whnf
\times \Pot(\Whnf)$ as follows \cite[Sec.~5.1]{abelCoquandDybjer:mpc08}:
\begin{gather*}
  \ru{}{(N,\Wne) \in \Univ i}
\qquad
  \rux{}{(\Set[j],|\Univ j|) \in \Univ i}{(\Set[j],\Set[i]) \in \Axiom}
\\[2ex]
  \rux{(U,\A) \in \close{\Univ i} \qquad
       \forall u \in \close\A.\, (\subst u x T, \F(u)) \in \close{\Univ j}
     }{(\efunT x U T, \PIAF) \in \Univ k
     }{(\Set[i],\Set[j],\Set[k]) \in \Rule}
\end{gather*}
Herein,
$\close{\Univ i} = \{ (T,\A) \mid (\whnf T,\A) \in \Univ i \}$
and $|\Univ j| = \{ A \mid (A,\A) \in \Univ j \mforsome \A \}$.
Only interested in computational strength, 
we treat relevant and irrelevant function spaces alike---at
the level of \emph{predicates} $\A$, irrelevance is anyhow not observable, 
only by \emph{relations} as given later.

The induction measure $A \in \Set[i]$ shall now mean the minimum
height of a derivation of $(A,\A) \in \Univ i$ for some $\A$.  Note that
due to universe stratification, $A \in \Set[i]$ is smaller than
$\Set[i] \in \Set[j]$.

\subsection{A Kripke Logical Relation}

Let $\Delta \der t \E t' \evarof T$ stand for the conjunction of the
propositions
\begin{itemize}
\item $\Delta \der t \evarof T$ and $\Delta \der t' \evarof T$, and
\item $\Delta \der t = t' \evarof T$.
\end{itemize}
By induction on
$A \in s$
we define two Kripke relations
\[
\begin{array}{r@{~}c@{~}l}
  \Delta \der A & \circleEq & A' : s
\\
  \Delta \der a & \circleEq & a' : A. 
\end{array}
\]
together with their respective closures $\cleq$ and the generalization
to $\evar$.
For better readability, the clauses are given in rule form
meaning that the conclusion \emph{is defined as} the conjunction of
the premises. $\forall$ and $\implies$ are meta-level quantification
and implication, respectively.
\begin{gather*}
  \ru{\Delta \der N   \E  N' : s
    }{\Delta \der N \Seq N' : s}
\qquad
  \ru{\Delta \der n  \E  n' : N
    }{\Delta \der n \Seq  n' : N}
\qquad
  \rux{\der \Delta 
     }{\Delta \der s \Seq  s : s'
     }{(s,s')}
\end{gather*}
\begin{gather*}
 \rux{\lcol{\Delta \der U \cleq U' : s_1} \\
      \forall \hatDelta \leq \Delta,~
      \hatDelta \der u \cleq u' \evarof U \implies
      \hatDelta \der \subst u x T \cleq \subst {u'}x{T'} : s_2 \\
      \Delta \der \efunTs x U {s_1,s_2} T  \E \efunTs x {U'}{s_1,s_2}{T'} : s_3
    }{\Delta \der \efunTs x U {s_1,s_2} T \Seq \efunTs x {U'}{s_1,s_2}{T'} : s_3
    }{(s_1,s_2,s_3)}
\end{gather*}
\begin{gather*}
  \ru{\forall \hatDelta \leq \Delta,~
      \hatDelta \der u \cleq u' \evarof U \implies
      \hatDelta \der f \eapp u  \cleq f' \eapp u' : \subst {u}x{T}
          \\
      \Delta \der f \E f' : \efunTs x {U} {s,s'} {T}
    }{\Delta \der f \Seq f' : \efunTs x {U} {s,s'} {T}}
\end{gather*}
\begin{gather*}
  \ru{T \evalsto A \qquad \Delta \der T = A \\
      t \evalsto a \qquad
      \Delta \der t = a : A \qquad
      \Delta \der t' = a' : A \qquad
      t' \evalsto a'  \\
      \Delta \der a \Seq a' : A \\
      \Delta \der t \E t' : T
    }{\Delta \der t \cleq t' : T}
\end{gather*}
\begin{gather*}
  \ru{\resurrect\Delta \der a \Seq a : A \qquad
      \resurrect{\Delta} \der a' \Seq a' : A
    }{\Delta \der a \Seq a' \div A}
\qquad
  \ru{\resurrect\Delta \der t \cleq  t : T \qquad
      \resurrect{\Delta} \der t' \cleq t' : T
    }{\Delta \der t \cleq t' \div T}
\end{gather*}
It is immediate that the logical relation contains only well-typed
and definitionally equal terms.  We will demonstrate that it 
is also closed under weakening and conversion, symmetric and transitive.

\begin{lem}[Weakening]\label{lem:weaklr} \bla
  \begin{enumerate}[\em(1)]
  \item If $\Delta \der a \Seq a' : A$ and
    $\hatDelta \leq \Delta$ then there exists 
    a derivation of $\hatDelta \der a \Seq a' : A$ with the
    same height.
  \item Analogously for $\Delta \der t \cleq t' : T$.
  \end{enumerate}
\end{lem}
\begin{proof}
  By induction on $A \in s$ and $T \in s$, resp.
\end{proof}

\begin{lem}[Type conversion]\label{lem:convlr} \bla
  \begin{enumerate}[\em(1)]
  \item
  If $\Gamma \der A \Seq A' : s$ then $\Gamma \der a \Seq a' : A$ 
  iff $\Gamma \der a \Seq a' : A'$.
  \item
  If $\Gamma \der T \cleq T' : s$ then $\Gamma \der t \cleq t' : T$
  iff $\Gamma \der t \cleq t' : T'$.
  \end{enumerate}
\end{lem}
\begin{proof}
  Simultaneously induction in $A \in s$ and $T \in s$, resp.  
  We show the ``if'' direction, the ``only if'' follows analogously.
  The interesting case is the one of functions.
\begin{enumerate}[\hbox to8 pt{\hfill}]                                        
\item\noindent{\hskip-12 pt\em Case\enspace} 
\begin{multline*}
 \rux{\lcol{\Delta \der U \cleq U' : s_1} \\
      \forall \hatDelta \leq \Delta,~
      \hatDelta \der u \cleq u' \evarof U \implies
      \hatDelta \der \subst u x T \cleq \subst {u'}x{T'} : s_2 \\
      \Delta \der \efunTs x U {s_1,s_2} T  \E \efunTs x {U'}{s_1,s_2}{T'} : s_3
    }{\Delta \der \efunTs x U {s_1,s_2} T \Seq \efunTs x {U'}{s_1,s_2}{T'} : s_3
    }{}
\\
  \ru{\forall \hatDelta \leq \Delta,~
      \hatDelta \der u \cleq u' \evarof U \implies
      \hatDelta \der f \eapp u  \cleq f' \eapp u' : \subst {u}x{T}
          \\
      \Delta \der f \E f' : \efunTs x {U} {s,s'} {T}
    }{\Delta \der f \Seq f' : \efunTs x {U} {s,s'} {T}}  
\end{multline*}
First, $\Delta \der f \E f' : \efunTs x {U'} {s,s'} {T'}$, holds because of
the conversion rule for typing and equality. 
Now assume arbitrary
$\hatDelta \leq \Delta$ and $\hatDelta \der u \cleq u'
\evarof {U'}$ and show $\hatDelta \der f \eapp u \cleq f' \eapp u' :
\subst {u} x {T'}$.
By induction hypothesis on $U \in s_1$ we have $\hatDelta \der u \cleq u'
\evarof U$, thus, $\hatDelta \der f \eapp u \cleq f' \eapp u' :
\subst u x T$ by assumption.  By induction hypothesis on $\subst u x
T \in s_2$ we obtain 
$\hatDelta \der f \eapp u \cleq f' \eapp u' : \subst u x {T'}$.
\end{enumerate}
\end{proof}

\begin{lem}[Symmetry and Transitivity]
  \label{lem:symlr}
  \label{lem:translr} \bla
  Let $\Delta \der T \cleq T : s$.
  \begin{enumerate}[\em(1)]
  \item
  If $\Delta \der t \cleq t' : T$
  then $\Delta \der t' \cleq t : T$.
  \item
  If $\Delta \der t_1 \cleq t_2 : T$ and
  $\Delta \der t_2 \cleq  t_3 : T$ then
  $\Delta \der t_1 \cleq t_3 : T$.
  \end{enumerate}
\end{lem}
\begin{proof}
  We generalize the two statements to whnfs $\Delta \der A \Seq A : s$ and
  prove all four statements
  simultaneously by induction in $A \in s$ and $T \in s$, resp.  
\begin{enumerate}[\hbox to8 pt{\hfill}]  
\item\noindent{\hskip-12 pt\em Case\enspace}  Let us look at the case for functions.
\[
 \rux{\lcol{\Delta \der U \cleq U : s_1} \\
      \forall \hatDelta \leq \Delta,~
      \hatDelta \der u \cleq u' \evarof U \implies
      \hatDelta \der \subst u x T \cleq \subst {u'}x{T} : s_2 \\
      \Delta \der \efunTs x U {s_1,s_2} T  \E \efunTs x {U}{s_1,s_2}{T} : s_3
    }{\Delta \der \efunTs x U {s_1,s_2} T \Seq \efunTs x {U}{s_1,s_2}{T} : s_3
    }{}
\]
\begin{enumerate}[\hbox to8 pt{\hfill}]

\item\noindent{\hskip-12 pt\em Case\enspace} Symmetry: 
\begin{displaymath}
  \ru{\forall \hatDelta \leq \Delta,~
      \hatDelta \der u \cleq u' \evarof U \implies
      \hatDelta \der f \eapp u  \cleq f' \eapp u' : \subst {u}x{T}
          \\
      \Delta \der f \E f' : \efunT x {U} {T}
    }{\Delta \der f \Seq f' : \efunT x {U} {T}}
\end{displaymath}
To show $\Delta \der f' \Seq f : \efunT x U T$, assume arbitrary
$\hatDelta \leq \Delta$ and $\hatDelta \der u' \cleq u
\evarof U$ and show $\hatDelta \der f' \eapp u' \cleq f \eapp u :
\subst {u'} x T$.
By induction hypothesis on $U \in s_2$, with 
weakened $\Gamma \der U \cleq U : s_1$, 
we have $\hatDelta \der u \cleq u'
\evarof U$, thus, $\hatDelta \der f \eapp u \cleq f' \eapp u' :
\subst u x T$ by assumption.  Using symmetry and transitivity on $U$
we obtain $\hatDelta \der u \cleq u \evarof U$, thus, 
$\hatDelta \der \subst u x T \cleq \subst {u} x T : s_2$. 
By induction hypothesis on $\subst u x
T \in s_2$ we apply symmetry to obtain
$\hatDelta \der f' \eapp u' \cleq f \eapp u : \subst u x T$, 
and since $\hatDelta \der \subst u x T \cleq \subst {u'} x T : s_2$ 
we conclude by type
conversion (Lemma~\ref{lem:convlr}).

\item\noindent{\hskip-12 pt\em Case\enspace} Transitivity:
\begin{multline*}
  \ru{\forall \hatDelta \leq \Delta,~
      \hatDelta \der u \cleq u' \evarof U \implies
      \hatDelta \der f_1 \eapp u  \cleq f_2 \eapp u' : \subst {u}x{T}
          \\
      \Delta \der f_1 \E f_2 : \efunTs x {U} {s,s'} {T}
    }{\Delta \der f_1 \Seq f_2 : \efunTs x {U} {s,s'} {T}}
\\[2ex]
  \ru{\forall \hatDelta \leq \Delta,~
      \hatDelta \der u \cleq u' \evarof U \implies
      \hatDelta \der f_2 \eapp u  \cleq f_3 \eapp u' : \subst {u}x{T}
          \\
      \Delta \der f_2 \E f_3 : \efunTs x {U} {s,s'} {T}
    }{\Delta \der f_2 \Seq f_3 : \efunTs x {U} {s,s'} {T}}
\end{multline*}

We wish to prove that
$\Delta \der f_1 \Seq f_3 : \efunTs x {U} {s,s'} {T}$.
We get $\Delta \der f_1 \E f_3 : \efunTs x U {s,s'} T$
immediately by transitivity of definitional equality.
Given $\Gamma \leq \Delta$ and
$\Gamma \der u \cleq u' \evarof U$,
we need to show that
$\Gamma \der f_1 \eapp u \cleq f_3 \eapp u' : \subst {u}xT$.

As $\hatDelta \der \_ \cleq \_ : U$ is a PER by induction hypothesis,
we have $\hatDelta \der u \cleq u \evarof U$, which entails
$f_1 \eapp u \cleq f_2 \eapp u : \subst {u} x T$. From
$\Gamma \der u \cleq u' \evarof U$ also have
$\Gamma \der f_2 \eapp u \cleq f_3 \eapp u' : \subst {u}x{T}$,
which allows to conclude 
$\Gamma \der f_1 \eapp u \cleq f_3 \eapp u' : \subst {u}xT$
by transitivity at $\subst {u}x{T}$.
\end{enumerate}

\item\noindent{\hskip-12 pt\em Case\enspace}  Now, we consider function spaces:

\begin{enumerate}[\hbox to8 pt{\hfill}]
\item\noindent{\hskip-12 pt\em Case\enspace} Transitivity:  
\begin{multline*}
 \rux{\lcol{\Delta \der U_1 \cleq U_2 : s_1} \\
      \forall \hatDelta \leq \Delta,~
      \hatDelta \der u \cleq u' \evarof U_1 \implies
      \hatDelta \der \subst u x {T_1} \cleq \subst {u'}x{T_2} : s_2 \\
      \Delta \der \efunTs x {U_1}{s_1,s_2}{T_1} 
        \E \efunTs x {U_2}{s_1,s_2}{T_2} : s_3
    }{\Delta \der \efunTs x {U_1}{s_1,s_2}{T_1}
        \Seq \efunTs x {U_2}{s_1,s_2}{T_2} : s_3
    }{}
\\[2ex]
 \rux{\lcol{\Delta \der U_2 \cleq U_3 : s_1} \\
      \forall \hatDelta \leq \Delta,~
      \hatDelta \der u \cleq u' \evarof U_2 \implies
      \hatDelta \der \subst u x {T_2} \cleq \subst {u'}x{T_3} : s_2 \\
      \Delta \der \efunTs x {U_2}{s_1,s_2}{T_2}
        \E \efunTs x {U_3}{s_1,s_2}{T_3} : s_3
    }{\Delta \der \efunTs x {U_2}{s_1,s_2}{T_2}
        \Seq \efunTs x {U_3}{s_1,s_2}{T_3} : s_3
    }{}
\end{multline*}

By transitivity we have
$\Delta \der \efunT x {U_1} {T_1} \E \efunT x {U_3} {T_3} : s_3$
and $\Delta \der U_1 \cleq U_3 : s_1$ by induction hypothesis on $s_1$.

Note that this is where the arrow sort annotations are useful.
Without them we would not know that the sorts in both derivations are equal.
We could have $\Delta \der U_1 \cleq U_2 : s_1$ and $\Delta \der U_2 \cleq U_3 : s_1'$
for apparently unrelated $s_1$ and $s_1'$, and would therefore be unable to use
transitivity.

Given $\Gamma \leq \Delta$ and $\Gamma \der u \cleq u' \evarof U_1$,
we need to show that $\Gamma \der \subst u x {T_1} \cleq \subst {u'} x
{T_3} : s_3$.
As $\cleq$ at type $U$ is a PER by induction hypothesis, 
we have $\Gamma \der u \cleq u \evarof U_1$,
from which we can deduce
$\Gamma \der \subst {u} x {T_1} \cleq \subst {u} x {T_2} : s_2$.
By conversion using $\Delta \der U_1 \cleq U_2 : s_1$ -- weakened at $\Gamma$ --
we have $\Gamma \der u \cleq u' \evarof U_2$, which implies
$\Gamma \der \subst u x {T_2} \cleq \subst {u'} x {T_3} : s_2$.
This allows us to conclude by transitivity at type $s_2$.

\end{enumerate}
\end{enumerate}
\end{proof}

In the following we show that the variables are in the logical
relation, \ie, $\Delta \der x \Seq x : \Delta(x)$ for well-formed
contexts $\Delta$.  As usual, this statement has to be generalized to
neutrals $n$ to be proven inductively.

\begin{lem}[Into the logical relation]\label{lem:intolr} Let $T \in s$.
  If\/ $\Delta \der n :=: n' \evarof T$
  then $\Delta \der n \cleq n' \evarof T$.
\end{lem}
\begin{proof}
By induction on $T \in s$.

\begin{enumerate}[\hbox to8 pt{\hfill}] 

\item\noindent{\hskip-12 pt\em Case\enspace} $N \in s$ and $\Delta \der n :=: n'
\evarof N$.  Then $\Delta \der n \Seq n' \evarof N$ by 
cases on $\evarof$, unfolding definitions.

\item\noindent{\hskip-12 pt\em Case\enspace} $s \in s'$
and $\Delta \der N :=: N' \evarof s$.
Then $\Delta \der N \Seq N' \evarof s$ by cases on $\evarof$.

\item\noindent{\hskip-12 pt\em Case\enspace}
  $\efunT x {U}{T} \in s_3$ and
  $\Delta \der n :=: n' \evarof_0 \efunT x U T$.

  First, the case for $\evar_0 = \noterased$.
  We have
  $\Delta \der n  \E n' : \efunT x {U} {T}$.  Assume arbitrary
  $\hatDelta \leq \Delta$ and
  $\hatDelta  \der u  \cleq u' \evarof U$,
  which yields
  $\hatDelta  \der u \E u' \evarof U$ and
  $\hatDelta  \der \subst u x T \cleq \subst{u}x{T} : s_2$.
  By weakening, 
  $\hatDelta \der n \eapp u \E n' \eapp u' : \subst u x T$,
  thus, by induction hypothesis,
  $\hatDelta \der n \eapp u \Seq n' \eapp u' : \subst u x T$,
  q.e.d. 

   The case for $\evar_0 = \erased$ proceeds analogously.
\end{enumerate}
\end{proof}








\subsection{Validity in the Model}



We now extend our logical relation $\cleq$ to substitutions, by
induction on the destination context.
\begin{gather*}
  \ru{}{\Delta \der \sigma \cleq \sigma' : \cempty}
\qquad
  \ru{\Delta \der \sigma \cleq \sigma' : \Gamma \qquad
      \Delta \der \sigma(x) \cleq \sigma'(x) \evarof {U}{\sigma}
    }{\Delta \der \sigma \cleq \sigma' : \eext {\Gamma}x{U}
    }
\end{gather*}
This relation inherits weakening from $\cleq$ for terms.

We then define the context ($\valid \Gamma$), type
($\Gamma \valid T = T'$) and term ($\Gamma \valid t = t' : T$) 
validity relations, by induction on the
length of contexts.

\begin{gather*}
  \ru{}{\valid \cempty}
\qquad
  \ru{ \valid \Gamma
       \quad
       \Gamma \valid U
    }{ \valid \eext \Gamma x U}
\qquad
\qquad
   \ru{\Gamma \valid T = T' : s
     }{\Gamma \valid T = T'}
\qquad
  \ru{\Gamma \valid T = T
    }{\Gamma \valid T}
\\[0.5em]
  \ru{\valid \Gamma \qquad
      (\Gamma \valid T \munless T = s)\\
      \forall \Delta,\sigma,\sigma', ~
        \Delta \der \sigma \cleq \sigma' : \Gamma
        \implies
        \Delta \der t \sigma \cleq {t'}{\sigma'} :  {T}{\sigma}
    }{\Gamma \valid t = t' : T}
\qquad
  \ru{\Gamma \valid t = t : T
    }{\Gamma \valid t : T}
\end{gather*}
Because of its asymmetric definition, the logical relation on
substitutions may not be a PER in general, but it is for valid
contexts.
\begin{lem}[Substitution relation is a PER]\label{lem:substlrper}
  If $\valid \Gamma$, then $\Delta \der \_ \cleq \_ : \Gamma$ is 
  symmetric and transitive.
\end{lem}
\begin{proof}
  By induction on $\Gamma$.  We demonstrate symmetry for the case
  $\valid \eext \Gamma x U$.
\[
  \ru{\Delta \der \sigma \cleq \sigma' : \Gamma \qquad
      \Delta \der \sigma(x) \cleq \sigma'(x) \evarof {U}{\sigma}
    }{\Delta \der \sigma \cleq \sigma' : \eext {\Gamma}x{U}
    }
\]
  By induction hypothesis, 
  $\Delta \der \sigma' \cleq \sigma : \Gamma$, and by symmetry of 
  $\cleq$ for terms (Lemma~\ref{lem:symlr}),
  $\Delta \der \sigma'(x) \cleq \sigma(x) \evarof U\sigma$.  
  We instantiate $\Gamma \valid U$ to 
  $\Delta \der U \sigma \cleq U \sigma' : s$ and conclude 
  $\Delta \der \sigma'(x) \cleq \sigma(x) \evarof U\sigma'$ by
  conversion (Lemma~\ref{lem:convlr}).
\end{proof}
\begin{lem}[Validity is a PER]
  The relation $\Gamma \valid \_ = \_ : T$ is symmetric and transitive.
\end{lem}
\begin{proof}
  Symmetry requires symmetry of $\cleq$ for substitutions and
  conversion with $\Delta \der T\sigma \cleq T \sigma' : s'$, similar
  as in Lemma~\ref{lem:substlrper}.

  We demonstrate transitivity in detail.
  Given $\Gamma \valid t_1 = t_2 : T$ and $\Gamma \valid t_2 = t_3 :
  T$ we show $\Gamma \valid t_1 = t_3 : T$.  Clearly, $\valid \Gamma$
  and $\Gamma \valid T$ or $T = s$ by one of our two
  assumptions. Assume arbitrary 
  $\Delta \der \sigma \cleq \sigma' : \Gamma$ and show
  $\Delta \der t_1\sigma \cleq t_3\sigma' : T\sigma$.
  By Lemma~\ref{lem:substlrper}, 
  $\Delta \der \sigma \cleq \sigma : \Gamma$,
  thus $\Delta \der t_1 \sigma \cleq t_2 \sigma : T\sigma$.  Also,
  $\Delta \der t_2 \sigma \cleq t_3 \sigma' : T\sigma$ which entails
  our goal by transitivity of $\cleq$ (Lemma~\ref{lem:translr}). 
\end{proof}

\begin{lem}[Function type injectivity is valid]
  \label{lem:funinjval}
  If\/ $\Gamma \valid \efunTs x U {s_1,s_2} T = \efunTs x {U'}{s_1',s_2'}{T'}$
  then $s_1 = s_1'$ and $s_2 = s_2'$ and $\Gamma \valid U = U' : s_1$
  and $\eext \Gamma x {U'} \valid T = T' : s_2$.
\end{lem}
\begin{proof}
  Assume arbitrary $\Delta \der \sigma \cleq \sigma' : \Gamma$.  We have
  $\Delta \der \efunTs x {U\sigma} {s_1,s_2} {T\sigma} \cleq 
               \efunTs x {U'\sigma'}{s_1',s_2'}{T'\sigma'} : s_3$,
  thus by definition $s_1 = s_1'$ and $s_2 = s_2'$ and
  $\Delta \der U'\sigma' \cleq U\sigma : s_1$---note that 
  sorts are closed and therefore invariant by substitution.
  By symmetry of $\cleq$, and since $\Delta,\sigma,\sigma'$ were arbitrary,
  we have $\Gamma \valid U = U' : s_1$.  

  Further, assume arbitrary $\Delta \der u \cleq u' \evarof U'\sigma$ and
  let $\rho = (\sigma,u/x)$ and $\rho' = (\sigma',u'/x)$.  Note that
  \wwlog, $x \not\in\dom(\Gamma)$ and $x \not \in \FV(U')$ and
  $\Delta \der \rho \cleq \rho' :  \eext \Gamma x {U'}$.  We have
  $\Delta \der T \rho \cleq T' \rho' : s_2$ and since
  $\rho,\rho'$ were arbitrary, $\eext\Gamma x{U'} \valid T = T' : s_2$.
\end{proof}



\newcommand{\GammaxU}{\eext \Gamma x U}
\newcommand{\GammaxUp}{\eext {\Gamma'}x{U'}}
\newcommand{\GammaxUpp}{\eext {\Gamma''}x{U''}}

\begin{lem}[Context satisfiable] 
  \label{lem:idsubst}
  If $\valid \Gamma$ then $\der \Gamma$ and
  $\Gamma  \der \sid \cleq \sid : \Gamma$.
\end{lem}
\begin{proof}
  By induction on $\Gamma$. The $\cempty$ case is immediate. 
  In the $\GammaxU$ case, given
  \[ \ru{  \valid \Gamma
           \qquad
           \Gamma \valid U
        }{
           \valid \GammaxU
        } 
  \]
  we can use inference
  \[ \ru{  \GammaxU \der \sid \cleq \sid : \Gamma
           \qquad
           \GammaxU \der \sid(x) \cleq \sid(x) \evarof {U}{\sid}
        }{
            \eext \Gamma x U \der \sid \cleq \sid : \eext {\Gamma}x{U}
        }
     .
  \]
  From the induction hypothesis 
  $\Gamma \der \sid \cleq \sid : \Gamma$, we obtain the first premise
  by weakening of $\cleq$.
  It also yields $\Gamma \der U\sid : s\sid$ for some $s$
  by definition of $\Gamma \valid U$.  
  Using induction hypothesis, $\der \Gamma$,
  this entails $\der \GammaxU$.  Further,
  $\GammaxU \der x = x \evarof U$, and since trivially
  $\GammaxU \der x \eqinf x \evarof U$,
  we can derive
  $\GammaxU \der x \cleq x \evarof U$,
  by the Lemma~\ref{lem:intolr}.
  This concludes the second premise
  $\GammaxU \der \sid(x) \cleq \sid(x) \evarof {U}{\sid}$.
\end{proof}

We can now show that every equation valid in the model is derivable in \IITT.
\begin{thm}[Completeness of \IITT rules]\label{thm:soundmodel}
    If $\Gamma \valid t = t' : T$ then 
    both $\Gamma \der t : T$ and $\Gamma \der t' : T$ and 
    $\Gamma \der t = t' : T$ and
    $\Gamma \der T$.
\end{thm}
\begin{proof}
  Using Lemma~\ref{lem:idsubst} we obtain $\Gamma \der t \cleq
  t' : T$, which entails $\Gamma \der t,t' : T$ and
  $\Gamma \der t = t' : T$.  
  Analogously, since our assumption entails $\Gamma
  \valid T$ by definition, we get $\Gamma \der T$.
\end{proof}

\subsection{Fundamental theorem}

We prove a series of lemmata which constitute parts of the fundamental
theorem for the Kripke logical relation.

\begin{lem}[Resurrection] \label{lem:resurrect}
  If $\valid \Gamma$ and
  $\Delta \der \sigma \cleq \sigma' : \Gamma$ then 
  $\resurrect\Delta \der \sigma \cleq \sigma : \resurrect\Gamma$ and
  $\resurrect\Delta \der \sigma' \cleq \sigma' : \resurrect\Gamma$.
\end{lem}
\begin{proof}
  By induction on $\Gamma$, the interesting case being
\[
  \ru{\Delta \der \sigma \cleq \sigma' : \Gamma \qquad
      \Delta \der \sigma(x) \cleq \sigma'(x) \evarof U\sigma
    }{\Delta \der \sigma \cleq \sigma' : \eext \Gamma x U}
  .
\]
  First, we show $\resurrect\Delta \der \sigma \cleq \sigma :
  (\cext{\resurrect \Gamma} x U)$.  By induction hypothesis
  $\resurrect\Delta \der \sigma \cleq \sigma : \resurrect\Gamma$,
  and by definition, 
  $\resurrect\Delta \der \sigma(x) \cleq \sigma(x) : U\sigma$.
  This immediately entails our goal.

  For the second goal  $\resurrect\Delta \der \sigma' \cleq \sigma' :
  (\cext{\resurrect \Gamma} x U)$, observe that
  $\Gamma \valid U$, hence $\Delta \der U\sigma \cleq U\sigma' : s$
  for some sort $s$.  Thus, we can cast our hypothesis
  $\Delta \der \sigma'(x) \cleq \sigma'(x) : U\sigma$ to $U\sigma'$
  and conclude analogously. 
\end{proof}

\begin{cor}\label{cor:relevance-polymorphic-instantiation}
  If $\eresurrect \Gamma \valid u : U$ and
  $\Delta \der \sigma \cleq \sigma' : \Gamma$
  then $\Delta \der u\sigma \cleq u\sigma' \evarof U\sigma$.
\end{cor}
\begin{proof}
  In case $\evar = \noterased$ it holds by definition,
  but we need resurrection for $\evar = \erased$.
  If $\Delta \der \sigma \cleq \sigma' : \Gamma$,
  then by resurrection (Lemma~\ref{lem:resurrect}) we have
  $\resurrect\Delta \der \sigma \cleq \sigma : \resurrect\Gamma$,
  so from $\resurrect \Gamma \valid u : U$ we deduce
  $\resurrect\Delta \der u\sigma \cleq u\sigma : U\sigma$.
  Analogously we get
  $\resurrect\Delta \der u\sigma' \cleq u\sigma' : U\sigma'$
  which we cast to
  $\resurrect\Delta \der u\sigma' \cleq u\sigma' : U\sigma$.
\end{proof}

\begin{lem}[Validity of $\beta$-reduction]
\[
 \nru{\reqbeta
    }{\eext \Gamma x U \valid t : T  \qquad
      \eresurrect \Gamma \valid u : U 
    }{\Gamma \valid (\elam x U t) \eapp u = \subst u x t : \subst u x T}
\]
\end{lem}
\begin{proof}
  $\valid \Gamma$ is contained in the first
  hypothesis $\Gamma \valid u \evarof U$.
  Then, given $\Delta \der \rho \cleq \rho' : \Gamma$ we need to show
  $\Delta \der (\elam x U t)\rho \eapp u\rho \cleq 
     t(\rho',u\rho'/x) : T(\rho,u\rho/x)$
  and also $\Delta \der \subst u x T \rho \cleq \subst u x T \rho' :
  s$ for some $s$ 
  (the latter to get $\Gamma \valid \subst u x T$).
  
  Let $\sigma = (\rho, u\rho/x)$ and $\sigma' = (\rho', u\rho'/x)$.
  From the second hypothesis and Cor.~
  \ref{cor:relevance-polymorphic-instantiation} we get
  $\Delta \der u\rho \cleq u\rho' \evarof U\rho$, which gives
  $\Delta \der \sigma \cleq \sigma' : \eext \Gamma x U$.
  By instantiating the first hypothesis we get
  $\Delta \der t\sigma \cleq t\sigma' : T\sigma$,
  and also (from the premise $\eext \Gamma x U \valid T$)
  $\Delta \der T\sigma = T\sigma'$,
  which gives $\Gamma \valid \subst u x T$.

  Finally, from $\Delta \der t\sigma \cleq t\sigma'$ we get the desired
  $\Delta \der (\elam x U t)\rho \eapp u\rho \cleq t\sigma' : T\sigma$,
  as $\cleq$ is closed by weak head expansion to well-typed
  $\Delta \der (\elam x U t)\rho \eapp u\rho : T\sigma$.
\end{proof}

\begin{lem}[Validity of $\eta$]
\[ \nru{\reqeta
    }{\Gamma \valid t : \efunT x U T 
    }{\Gamma \valid t = \elam x U {t \eapp x} : \efunT x U T
    }
\]  
\end{lem}
\begin{proof}
$\valid \Gamma$ and $\Gamma \valid \efunT x U T$ are direct consequences
of our hypothesis. Given $\Delta \der \rho \cleq \rho' : \Gamma$, we need to show
$\Delta \der t\rho \cleq (\elam x U {t \eapp x}) \rho' : (\efunT x U T)\rho$.
\Wlog, $x$ is not free in the domain nor range of substitutions $\rho$
and $\rho'$, thus with $t' := t\rho$, $t'' := t\rho'$,
$U' := U\rho$, $U'' := U\rho'$, $T' := T\rho$ and $T'' := T\rho'$
it is sufficient to show
$\Delta \der t' \cleq \elam x {U''} {t'' \eapp x} : \efunT x {U'} {T'}$.

First, given $(\Delta', u, u')$ such that
$\Delta' \leq \Delta$ and $\Delta' \der u \cleq u' \evarof U'$, we show
$\Delta' \der t' \eapp u
  \cleq (\elam x{U''} {t'' \eapp x}) \eapp u' : \subst u x {T'}$.
Our hypothesis $\Gamma \valid t : \efunT x U T$ entails
$\Delta \der t\rho \cleq t\rho' : (\efunT x U T)\rho$,
that is to say $\Delta \der t' \cleq t'' : \efunT x {U'}{T'}$.
This logical relation at a function type,
when instantiated to $(\Delta', u, u')$,
gives us $\Delta' \der t' \eapp u \cleq t'' \eapp u' : \subst u x {T'}$,
which weak-head expands to the desired goal.

Second, we show
$\Delta \der t' \E \elam x {U''} {t'' \eapp x} : \efunT x {U'} {T'}$.
\begin{itemize}
\item $\Gamma \der t' : \efunT x {U'} {T'}$ is a simple consequence of our hypothesis $\Gamma \valid t : \efunT x U T$.
\item $\Gamma \der \elam x {U''} {t'' \eapp x} : \efunT x {U'} {T'}$ has the following proof:
\[
\begin{prooftree}
  \[ 
  \[
    \[
      \[ \Gamma \valid t : \efunT x U T
         \Justifies 
         \Delta \der t'' : \efunT x {U''} {T''}
      \]
      \justifies
      \eext \Delta x {U''} \der t'' : \efunT x {U''} {T''}
      \using\text{weak} 
    \]
    \[ \justifies \eresurrect{(\eext \Delta x {U''})} \der x : U'' \using \text{var}\]
    \justifies 
    \eext {\Delta} x {U''} \der t'' \eapp x : T'' 
  \]
  \[
    \Gamma \valid \efunT x U T
    \Justifies \Delta \der \efunT x {U''} {T''}
  \]
  \justifies
  \Delta \der \elam x {U''} {t'' \eapp x} : \efunT x {U''} {T''}
  \]
  \justifies
  \Delta \der \elam x {U''} {t'' \eapp x} : \efunT x {U'} {T'}
  \using\text{conv}
\end{prooftree}
\]
\item $\Delta \der t' = \elam x {U''} {t'' \eapp x} : \efunT x {U'} {T'}$.
The $\eta$-rule of definitional equality gives us
$\Delta \der t'' = \elam x {U''} {t'' \eapp x} : \efunT x {U''} {T''}$.
From $\Gamma \valid \efunT x U T$ we can convert it to the type
$\efunT x {U'} {T'}$, and then conclude by transitivity using
$\Delta \der t' = t'' : \efunT x {U'} {T'}$, which is a direct consequence of
$\Gamma \valid t : \efunT x U T$.
\end{itemize}
\end{proof}

\begin{lem}[Validity of function equality]
\[ \nru{\reglambda
    }{\Gamma \valid U = U' \qquad \eext \Gamma x U \valid t = t' : T
    }{\Gamma \valid (\elam x U t) = (\elam x {U'} {t'}) : \efunT x U T
    }
\]
\begin{proof}
Again $\valid \Gamma$ and $\Gamma \valid \efunT x U T$ are simple consequences
of our hypotheses. Given $\Delta \der \rho \cleq \rho' : \Gamma$
(\wwlog, $x$ is not free in $\rho,\rho'$ domain or range),
we need to show
$\Delta \der (\elam x {U\rho} {t\rho}) \cleq (\elam x {U'\rho'} {t'\rho'})
  : (\efunT x {U\rho} {T\rho})$. We will skip the proof of
$\Delta \der (\elam x {U\rho} {t\rho}) \E (\elam x {U'\rho'} {t'\rho'})
  : (\efunT x {U\rho} {T\rho})$, as it is similar to the corresponding
part of the $\eta$-validity lemma.

Given $(\Delta', u, u')$ such that
$\Delta' \leq \Delta$ and $\Delta' \der u \cleq u' \evarof U\rho$,
we have to show that
$\Delta' \der (\elam x {U\rho} {t\rho}) \eapp u
  \cleq (\elam x {U'\rho'} {t'\rho'}) \eapp u'
  : \subst u x {T\rho}$.
Let $\sigma = (\rho, u/x)$ and $\sigma' = (\rho', u'/x)$.
As we supposed $\Delta' \der u \cleq u' \evarof U\rho$, we have
$\Delta' \der \sigma \cleq \sigma' : \eext \Gamma x {U\rho}$.
Instantiating the second hypothesis with $\Delta', \sigma, \sigma'$
therefore gives us $\Delta' \der t\sigma = t'\sigma' : T\sigma$,
which can also be written
$\Delta' \der \subst u x {t\rho}
  \cleq \subst {u'} x {t'\rho'} : \subst u x {T\rho}$,
which is weak-head expansible to our goal.
\end{proof}
\end{lem}

\begin{lem}[Validity of irrelevant application]
\[
  \nru{\reqerapp
     }{\Gamma \valid t = t' \of \erfunT x U T \qquad
       \resurrect\Gamma \valid u : U \qquad
       \resurrect{\Gamma} \valid u' : U
     }{\Gamma \valid t \erapp u = t' \erapp u' : \subst u x T
     }
\]
\end{lem}
\begin{proof}
Assume arbitrary $\Delta \der \rho \cleq \rho' : \Gamma$ and show 
$\Delta \der t\rho \erapp u\rho \cleq t'\rho' \erapp u'\rho' : 
 T (\rho, u\rho/x)$.  By the first hypothesis, it is sufficient to
 show
$\Delta \der u\rho \cleq u'\rho' \erof U\rho$, which means
$\resurrect\Delta \der u\rho \cleq u\rho : U\rho$ and
$\resurrect\Delta \der u'\rho' \cleq u'\rho' : U\rho$.
By Resurrection (Lemma~\ref{lem:resurrect}), 
$\resurrect\Delta \der \rho \cleq \rho : \resurrect\Gamma$, hence
$\resurrect\Delta \der u\rho \cleq u\rho : U\rho$ from the second
hypothesis.  Analogously, we obtain 
$\resurrect\Delta \der u'\rho' \cleq u'\rho' : U\rho'$ from the third
hypothesis which we can cast to $U\rho$ by virtue of $\Gamma \valid U$
which we get from $\Gamma \valid \erfunT x U T$ by
Lemma~\ref{lem:funinjval}.
\end{proof}

\begin{thm}[Fundamental theorem of logical relations] 
  \label{thm:compl} \bla
  \begin{enumerate}[\em(1)]
  \item If $\der \Gamma$ then $\valid \Gamma$.
  \item If $\Gamma \der t : T$ then $\Gamma \valid t : T$.
  \item If $\Gamma \der t = t' : T$ then
        $\Gamma \valid t = t' : T$.
  \end{enumerate}
\end{thm}
\begin{proof}
  By induction on the derivation.
\end{proof}
As a simple corollary we obtain syntactic validity, namely that
definitional equality implies well-typedness and well-typedness
implies well-formedness of the involved type.  This lemma could have
been proven purely syntactically, but the
syntactic proof requires a sequence of carefully arranged lemmata like
context conversion, substitution, functionality, and inversion on
types \cite{harperPfenning:equivalenceLF,abelCoquand:fundinf07}.
Our ``sledgehammer'' \emph{semantic} argument is built into the Kripke
logical relation, in the spirit of Goguen \cite{goguen:types00}.
\begin{cor}[Syntactic validity]\label{cor:synval} \bla
  \begin{enumerate}[\em(1)]
  \item If $\Gamma \der t : T$ then $\Gamma \der T$.  
  \item
    If $\Gamma \der t = t' : T$ then
    $\Gamma \der t : T$ and $\Gamma \der t' : T$.
  \end{enumerate}  
\end{cor}
\begin{proof}
  By the fundamental theorem, $\Gamma \der t = t' : T$ implies $\Gamma
  \valid t = t' : T$, which by Thm.~\ref{thm:soundmodel} implies
  $\Gamma \der t, t' : T$ and $\Gamma \der T$.
\end{proof}

\clearpage

\section{Meta-theoretic Consequences of the Model Construction}
\label{sec:meta}


In this section, we explicate the results established by the Kripke model.

\subsection{Admissibility of Substitution}

\newcommand{\Gprim}{\Gamma'}

Goguen \cite{goguen:types00} observes that admissibility of substitution
for the syntactic judgements can be inherited from the Kripke logical
relation, which is closed under substitution by its very definition.

To show that the judgements of \IITT are closed
under substitution we introduce relations
$\Gamma \der \sigma : \Gprim$ for substitution typing and
$\Gamma \der \sigma = \sigma' : \Gprim$
for substitution equality which
are given inductively by the following rules:
\begin{gather*}
  \ru{\der \Gamma
    }{\Gamma \der \sigma : \cempty}
\qquad
  \ru{\Gamma \der \sigma : \Gprim
       \qquad
      \Gprim \der U \qquad
      \Gamma  \der \sigma(x) \evarof U \sigma
    }{\Gamma \der \sigma  : \eext \Gprim x U}
\\[2ex]
  \ru{\der \Gamma
    }{\Gamma \der \sigma = \sigma' : \cempty}
\qquad
  \ru{\Gamma \der \sigma = \sigma' : \Gprim
       \qquad
      \Gprim \der U \qquad
      \Gamma  \der \sigma(x) = \sigma'(x) \evarof U \sigma
    }{\Gamma \der \sigma = \sigma' : \eext \Gprim x U}
\end{gather*}
Substitution typing and equality are closed under weakening.

Semantically, substitutions are explained by environments.  We define
substitution validity as follows, again in rule form but not
inductively:
\begin{gather*}
  \ru{\Gamma \valid \sigma = \sigma : \Gamma'
    }{\Gamma \valid \sigma : \Gamma'}
\qquad
  \ru{\valid \Gamma \qquad \valid \Gamma' \\
      \forall \Delta \valid \rho \cleq \rho' : \Gamma.~\
         \Delta \valid \sigma\rho \cleq \sigma'\rho' : \Gamma'
    }{\Gamma \valid \sigma = \sigma' : \Gamma'}
\end{gather*}
\begin{lem}[Fundamental lemma for substitutions]
  \label{lem:fundsubst} \bla 
  \begin{enumerate}[\em(1)]
  \item \label{it:jsubst} 
    If\/ $\Gamma \der \sigma : \Gamma'$ then $\Gamma \valid \sigma
    : \Gamma'$. 
  \item \label{it:jsubsteq}
   If\/ $\Gamma \der \sigma = \sigma' : \Gamma'$ then $\Gamma
    \valid \sigma = \sigma' : \Gamma'$. 
  \end{enumerate}
\end{lem}

\proof
  We demonstrate \ref{it:jsubsteq} by induction on $\Gamma \der \sigma
  = \sigma' : \Gamma'$.
\begin{enumerate}[\hbox to8 pt{\hfill}]       

\item\noindent{\hskip-12 pt\em Case\enspace}
\[
  \ru{\der \Gamma
    }{\Gamma \der \sigma = \sigma' : \cempty}  
\]
We have $\valid \Gamma$ by Thm.~\ref{thm:compl} and $\valid \cempty$
trivially.  Also, $\Delta \der \sigma \rho \cleq \sigma' \rho' : \cempty$
trivially for any $\Delta \der \rho \cleq \rho' : \Gamma$.

\item\noindent{\hskip-12 pt\em Case\enspace}
\[
  \ru{\Gamma \der \sigma = \sigma' : \Gprim
       \qquad
      \Gprim \der U \qquad
      \Gamma  \der \sigma(x) = \sigma'(x) \evarof U \sigma
    }{\Gamma \der \sigma = \sigma' : \eext \Gprim x U}
\]
We have $\valid \Gamma$ and $\valid \Gamma'$ by induction hypothesis
and $\Gamma' \valid U$ by Thm.~\ref{thm:compl}, thus, $\valid \eext
{\Gamma'} x U$.  Now assume arbitrary $\Delta \der \rho \cleq \rho' :
\Gamma$ and show $\Delta \der \sigma\rho \cleq \sigma'\rho' :
\eext{\Gamma'}xU$.   First, $\Delta \der \sigma \rho \cleq \sigma'
\rho' : \Gamma'$ follows by induction hypothesis.  The second subgoal
$\Delta \der (\sigma\rho)(x) \cleq (\sigma'\rho')(x) \evarof U\sigma\rho$
is just an instance of the second induction hypothesis.\qed\medskip
\end{enumerate}

\begin{thm}[Substitution and functionality] \label{thm:subst}
  \bla
  \begin{enumerate}[\em(1)]
   \item If $\Gamma \der \sigma : \Gprim$ and $\Gprim \der t : T$ then
     $\Gamma \der t \sigma : T \sigma$.
   \item
     If $\Gamma \der \sigma : \Gprim$.
     and $\Gprim \der t = t' : T$ then
     $\Gamma \der t \sigma = t' \sigma : T \sigma$.
   \item
     If $\Gamma \der \sigma = \sigma' : \Gprim$.
     and $\Gprim \der t : T$ then
     $\Gamma \der t \sigma = t \sigma' : T \sigma$.
   \item \label{it:substeq}
     If $\Gamma \der \sigma = \sigma' : \Gprim$.
     and $\Gprim \der t = t' : T$ then
     $\Gamma \der t \sigma = t' \sigma' : T \sigma$.
  \end{enumerate}
\end{thm}
\begin{proof}
  We demonstrate \ref{it:substeq}, the other cases are just
  variations of the theme.
  First, from $\Gamma \der \sigma = \sigma' : \Gamma'$ we get
  $\Gamma \der \sigma \cleq \sigma' : \Gamma'$ 
  by the fundamental lemma for substitutions 
  (Lemma~\ref{lem:fundsubst}), using the identity environment
  $\Gamma \der \sid \cleq \sid : \Gamma$.
  Now, by the fundamental theorem on $\Gamma \der t = t' : T$ we
  obtain $\Gamma \der t \sigma \cleq t' \sigma' : T \sigma$, which
  entails our goal $\Gamma \der t \sigma = t' \sigma' : T \sigma$
  by Thm.~\ref{thm:soundmodel}.
\end{proof}

\subsection{Context conversion}

Context equality $\der \Gamma = \Gamma'$ is defined inductively by the
rules
\begin{gather*}
  \ru{}{\der \cempty = \cempty}
\qquad
  \ru{\der \Gamma = \Gamma' \qquad
      \Gamma \der U = U'
    }{\der \eext \Gamma x U = \eext{\Gamma'}x{U'}}
  .
\end{gather*}

All declarative judgements are closed under context conversion.  This
fact is easy to prove by induction over derivations, but we get it
as just a special case of substitution.
\begin{lem}[Identity substitution]
  If $\der \Gamma = \Gamma'$ then $\Gamma \der \sid = \sid : \Gamma'$.
\end{lem}

\proof
  By induction on $\der \Gamma = \Gamma'$.
\begin{enumerate}[\hbox to8 pt{\hfill}]
\item\noindent{\hskip-12 pt\em Case\enspace}
\[
  \ru{\der \Gamma = \Gamma' \qquad
      \Gamma \der U = U'
    }{\der \eext \Gamma x U = \eext{\Gamma'}x{U'}}
\]
By induction hypothesis and weakening, 
$\eext \Gamma x U \der \sid = \sid : \Gamma'$.
Also, $\eext \Gamma x U \der x = x \evarof U$ and by conversion
$\eext \Gamma x U \der x = x \evarof U'$.  Together,
$\eext \Gamma x U \der \sid = \sid : \eext {\Gamma'}x{U'}$.\qed\medskip  
\end{enumerate}

\begin{thm}[Context conversion]\label{thm:cxtconv}
  Let $\der \Gamma' = \Gamma$.
  \begin{enumerate}[\em(1)]
  \item If $\Gamma \der t : T$ then $\Gamma' \der t : T$.
  \item If $\Gamma \der t = t' : T$ then $\Gamma' \der t = t' : T$.
  \end{enumerate}
\end{thm}
\begin{proof}
  By Thm.~\ref{thm:subst} with $\Gamma' \der \sid = \sid : \Gamma$.
\end{proof}
As a consequence, context equality is symmetric and transitive (we can
trade $\Gamma \der U = U'$ for $\Gamma' \der U = U'$).  Thus, context
conversion can be applied in the other direction as well.

\subsection{Inversion, injectivity, and type unicity}

A condition for the decidability of type checking is the ability to
invert typing derivations.  The proof requires substitution.
\begin{lem}[Inversion]\label{lem:inv} \bla
  \begin{enumerate}[\em(1)]
  \item If $\Gamma \der x : T$ then $(x \of U) \in \Gamma$ for some
    $U$ with $\Gamma \der U = T$.
  \item If $\Gamma \der \elam x U t : T$ then $\eext \Gamma x U \der t
    : T'$ for some $T'$ with $\Gamma \der \efunT x U {T'} = T$.
  \item If $\Gamma \der t \eapp u : T$ then $\Gamma \der t : \efunT x
    U {T'}$ and $\Gamma \der u \evarof U$ for some $U,T'$ with $\Gamma
    \der \subst u x {T'} = T$.
  \item If $\Gamma \der s : T$ then there is $(s,s') \in \Axiom$ such
    that $\Gamma \der s' = T$.
  \item If $\Gamma \der \efunTs x U {s_1,s_2} {T'} : T$ then $\Gamma \der U :
    s_1$ and $\eext \Gamma x U \der T' : s_2$, and for some $s_3$ we have
    $\Gamma \der s_3 = T$ and $(s_1,s_2,s_3) \in \Rule$.
  \end{enumerate}
\end{lem}
\begin{proof}
  Each by induction on the typing derivation.
\end{proof}
\begin{rem}
  The need for inversion during type checking is the only good reason
  to have separate typing rules and not simply define typing
  $\Gamma \der t : T$ as the diagonal 
  $\Gamma \der t = t : T$ of equality. 
  While by a logical relation argument
  we will obtain a suitable inversion result for 
  $\Gamma \der \efunT x U T = \efunT x U T$---the famous 
  \emph{function type injectivity} (Theorem~\ref{thm:funinj})---
  it seems hard to get something similar for application $t\,u$.
\end{rem}

Injectivity for function types \wrt\ typed equality
is known to be tricky.  It is connected
to subject reduction and required for many meta-theoretic results.
We harvest it from our Kripke model.
\begin{thm}[Function type injectivity]\label{thm:funinj}
  If $\Gamma \der \efunTs x U {s_1,s_2} T = \efunTs x{U'}{s_1',s_2'}{T'} : s_3$
  then $s_1 = s_1'$ and $s_2 = s_2'$ and $\Gamma \der U = U' : s_1$
  and  $\eext \Gamma x U \der T = T' : s_2$.
\end{thm}
\begin{proof}
  This follows from Lemma~\ref{lem:funinjval}.  Or we can prove it
  directly as follows:
  Since $\Gamma \der \sid \cleq \sid : \Gamma$
  we have by the fundamental theorem
  $\Gamma \der \efunTs x U {s_1,s_2} T \cleq \efunTs x{U'}{s_1,s_2}{T'} : s_3$
  which by inversion yields first $s_1 = s_1'$ and $s_2 = s_2'$ and
  $\Gamma \der U \cleq U' : s_1$ and
  $\Gamma \der U = U' : s_1$. Since
  $\eext \Gamma x U \der x \cleq x \evarof U$, we also obtain
  $\eext \Gamma x U \der T \cleq T' : s_2$ and conclude
  $\eext \Gamma x U \der T = T' : s_2$.
\end{proof}

From the inversion lemma we can prove uniqueness of types, since we
are dealing with a functional PTS, and we have function type injectivity.
\begin{thm}[Type unicity]\label{thm:tyuniq}
  If $\Gamma \der t : T$ and $\Gamma \der t : T'$ then 
  $\Gamma \der T = T'$.
\end{thm}
\begin{proof}
  By induction on $t$, using inversion.
\end{proof}

\subsection{\uppercase{N}ormalization and Subject Reduction}

An immediate consequence of the model construction is that each term
has a weak head normal form and that typing and equality is preserved
by weak head normalization.
\begin{thm}[Normalization and subject reduction]\label{thm:norm}
  If $\Gamma \der t : T$ then $t \evalsto a$ and $\Gamma \der t = a : T$.
\end{thm}
\begin{proof}
  By the fundamental theorem,
  $\Gamma \der t \cleq t : T$
  which by definition contains a derivation of
  $\Gamma \der t = \whnf t : T$.
\end{proof}

\subsection{Consistency}

Importantly, not every type is inhabited in \IITT, thus, it can be
used as a logic.  A prerequisite is that types can be distinguished,
which follows immediately from the construction of the logical relation.
\begin{lem}[Type constructor discrimination]\label{lem:tycondist}
  \bla
  Neutral types, sorts and function types are mutually unequal.
\begin{enumerate}[\em(1)]
\item $\Gamma \der N \not= s$.
\item $\Gamma \der N \not= \efunT x U T$.
\item $\Gamma \der s = s'$ implies $s \equiv s'$.
\item $\Gamma \der s \not= \efunT x U T$.
\end{enumerate}
\end{lem}
\begin{proof}
  By the fundamental theorem applied to the identity substitution.
  For instance, assuming $\Gamma \der N = s : s'$ we get 
  $\Gamma \der N \Seq s : s'$ but this is a contradiction to the
  definition of $\Seq$.
\end{proof}
From normalization and type constructor discrimination we can show that
not every type is inhabited.
\begin{thm}[Consistency] \label{thm:consistency}
  $X \of \Set[0] \not\der t : X$.
\end{thm}
\begin{proof}
  Let $\Gamma = (X \of \Set[0])$.
  Assuming $\Gamma \der t : X$, we have $\Gamma \der a : X$
  for the whnf $a$ of $t$.  We invert on the typing of $a$.
  By Lemma~\ref{lem:tycondist}, $X$ cannot
  be equal to a function type or sort, thus, $a$ can neither be a
  $\lambda$ nor a function type nor a sort, it can only be neutral.
  The only variable $X$ must be in the head of $a$, but since $X$ is
  not of function type, it cannot be applied.
  Thus, $a \equiv X$ and $\Gamma \der X : X$, implying $\Gamma \der X
  = \Set[0]$ by inversion (Lemma~\ref{lem:inv}).  
  This is in contradiction to Lemma~\ref{lem:tycondist}!
\end{proof}

\subsection{Soundness of Algorithmic Equality}

Soundness of the equality algorithm is a consequence of subject
reduction.
\begin{thm}[Soundness of algorithmic equality] \label{thm:aleqsound} \bla
\begin{enumerate}[\em(1)]
\item
  Let $\Delta \der t, t'  : T$
  . If
  $\Delta  \der t \eqchkt t' : T$ then
  $\Delta \der t = t' : T$.
 \item Let $\Delta \der n, n' : T$ 
  . If
  $\Delta \der n \eqinft n' : U$ then
  $\Delta \der n = n' : U$ and
  $\Delta \der U = T$.
\end{enumerate}
\end{thm}

\proof
  Generalize the theorem to all six algorithmic equality judgments and
  prove it by induction on the algorithmic equality derivation.  Since
  we have subject reduction, the proof proceeds mechanically, because
  each algorithmic rule corresponds, modulo weak head normalization,
  to a declarative rule.
  \begin{enumerate}[\hbox to8 pt{\hfill}] 

  \item\noindent{\hskip-12 pt\em Case\enspace} $\Delta \der T : s$ and $\Delta' \der T' : s$ and
  \[
    \ru{\Delta \der \whnf T \eqty \whnf T'
      }{\Delta \der T \eqtyt T'}
  \]
  By induction hypothesis,
  $\Delta  \der \whnf T = \whnf T' : s$.
  By subject reduction $\Delta \der T = \whnf T : s$ and
  $\Delta \der T' = \whnf T' : s$.  By transitivity
  $\Delta \der T =  T' : s$.

  \item\noindent{\hskip-12 pt\em Case\enspace}
  \[
  \ru{\Delta \der T \eqtyt T'
    }{\Delta \der T \eqchk T' : s}
  \]
  By induction hypothesis, $\Delta \der T = T' : s$.\qed
  \end{enumerate}

\subsection{Symmetry and Transitivity of Algorithmic Equality}

Since algorithmic equality is sound for well-typed terms, it is also
symmetric and transitive.

\begin{lem}[Type and context conversion in algorithmic equality]
\label{lem:tyconvalg} Let $\der \Delta = \Delta'$.
\begin{enumerate}[\em(1)]
\item If $\Delta \der A,A'$ and
      $\Delta \der A \eqty A'$ then
      $\Delta' \der A \eqty A'$.
\item If $\Delta \der n,n' : A$ and
      $\Delta \der n \eqinf n' : A$ then
      $\Delta' \der n \eqinf n' : A'$ for some
      $A'$ with $\Delta \der A = A'$.
\item If $\Delta \der t,t' : A$ and 
      $\Delta \der t \eqchk t' : A$ and
      $\Delta \der A = A'$ then
      $\Delta' \der t \eqchk t' : A'$.
\end{enumerate}
\end{lem}

\proof
  By induction on the derivation of algorithmic equality,
  where we extend the statements to $\eqinftr$ and $\eqchktr$
  accordingly.
  \begin{enumerate}[(1)]
  \item Type equality.
    \begin{enumerate}[\hbox to8 pt{\hfill}] 

\item\noindent{\hskip-12 pt\em Case\enspace}
\[
\ru{\Delta \der U \eqtyt U' \qquad
    \eext \Delta x U \der T \eqtyt T'
  }{\Delta \der \efunT x U T \eqty \efunT x {U'}{T'}}      
\]
By inversion, $\Delta \der U,U'$ and by induction hypothesis, 
$\Delta' \der U \eqtyt U'$.  Again by inversion, $\eext \Delta x U
\der T$ and $\eext \Delta x {U'} \der T'$, yet by soundness of
algorithmic equality, $\Delta \der U = U'$, hence $\eext \Delta x U
\der T'$ by context conversion.  Further, $\der \eext \Delta x U =
\eext {\Delta'} x {U}$.  Thus, we can apply the other induction
hypothesis to obtain $\eext {\Delta'} x {U} \der T \eqtyt T'$, which
finally yields $\Delta' \der \efunT x U T \eqty \efunT x {U'}{T'}$.
    \end{enumerate}

  \item Structural equality.
\begin{enumerate}[\hbox to8 pt{\hfill}] 

\item\noindent{\hskip-12 pt\em Case\enspace}
\[
  \ru{(x \of T) \in \Delta 
    }{\Delta \der x \eqinft x : T}
\]
  Since $\der \Delta = \Delta'$, there is a unique
  $(x \of T') \in \Delta'$ with $\Delta \der T = T'$.  Hence,
  $\Delta' \der x \eqinft x : T'$.

\item\noindent{\hskip-12 pt\em Case\enspace}

  \item Type-directed equality.
    \begin{enumerate}[\hbox to8 pt{\hfill}] 

\item\noindent{\hskip-12 pt\em Case\enspace} $\Delta \der t,t' : T$ and $\Delta \der T = T'$ and
\[
  \ru{T \evalsto A \qquad \Delta \der t \eqchk t' : A
    }{\Delta \der t \eqchkt t' : T}
\]      
    By normalization, $T' \evalsto A'$, and subject reduction $\Delta
    \der A = T = T' = A'$.  Since by conversion, $\Delta \der t,t' : A$, 
    by induction hypothesis $\Delta' \der t \eqchk t' : A'$.
    Thus, $\Delta' \der t \eqchkt t' : T'$.

\item\noindent{\hskip-12 pt\em Case\enspace} $\Delta \der \efunT x U T = A'$ and
\[
  \ru{\eext \Delta x U \der t \eapp x \eqchkt t' \eapp x : T
    }{\Delta \der t \eqchk t' : \efunT x U T}
\]
    By injectivity $A' \equiv \efunT x {U'}{T'}$ with $\Delta \der U =
    U'$ and $\eext \Delta x U \der T = T'$.  Since $\der \eext \Delta
    x U = \eext {\Delta'} x {U'}$, by induction hypothesis we have
    $\eext {\Delta'} x {U'} \der t \eapp x \eqchkt t' \eapp x : T'$.
    We conclude $\Delta' \der t \eqchk t' : \efunT x {U'}{T'}$.\qed
    \end{enumerate}
  \end{enumerate}
\end{enumerate}

\begin{lem}[Algorithmic equality is transitive]
  \label{lem:aleqtrans}
  Let $\der \Delta = \Delta'$.  In the following, let the terms
  submitted to algorithmic equality be well-typed.
  \begin{enumerate}[\em(1)]
  \item If $\Delta  \der n_1 \eqinft n_2 : T$ 
       and $\Delta' \der n_2 \eqinft n_3 : T'$ 
      then $\Delta  \der n_1 \eqinft n_3 : T$ and $\Delta \der T = T'$.
  \item If $\Delta  \der t_1 \eqchkt t_2 : T$
       and $\Delta' \der t_2 \eqchkt t_3 : T'$ and $\Delta \der T = T'$
      then $\Delta  \der t_1 \eqchkt t_3 : T$.
  \item If $\Delta  \der T_1 \eqtyt T_2 : s$ 
       and $\Delta' \der T_2 \eqtyt T_3 : s$
      then $\Delta  \der T_1 \eqtyt T_3 : s$  
  \end{enumerate}
\end{lem}

\proof
We extend these statements to $\eqinfr$ and $\eqchkr$
and prove them simultaneously by induction on the first derivation.
\begin{enumerate}[\hbox to8 pt{\hfill}]   

\item\noindent{\hskip-12 pt\em Case\enspace}
\[
\ru{\Delta \der n_1 \eqinft n_2 : T
  }{\Delta \der n_1 \eqchk  n_2 : N}
\qquad
\ru{\Delta' \der n_2 \eqinft n_3 : T'
  }{\Delta' \der n_2  \eqchk n_3 : N'}
\]
By induction hypothesis $\Delta \der n_1 \eqinft n_3 : T$, hence,
$\Delta \der n_1 \eqchk n_3 : N$.

\item\noindent{\hskip-12 pt\em Case\enspace}
\[
\ru{\Delta \der N_1 \eqinft N_2 : T
  }{\Delta \der N_1 \eqchk N_2}
\qquad
\ru{\Delta \der N_2 \eqinft N_3 : T'
  }{\Delta \der N_2 \eqchk N_3}
\]
Analogously.

\item\noindent{\hskip-12 pt\em Case\enspace}
\begin{multline*}
  \ru{\Delta \der n_1 \eqinf  n_2 : \funTs x {U}{s_1,s_2}{T} \qquad
      \Delta \der u_1 \eqchkt u_2 : U
    }{\Delta \der n_1 \app u_1 \eqinft n_2 \app u_2 : \subst {u_1}x{T}
    }
\\[0.7em]
  \ru{\Delta' \der n_2 \eqinf n_3 : \funTs x {U'}{s_1',s_2'}{T'} \qquad
      \Delta' \der u_2 \eqchkt u_3 : U'
    }{\Delta' \der n_2 \app u_2 \eqinft n_3 \app u_3 : \subst {u_2}x{T'}
    }
\end{multline*}
By induction hypothesis we have
$\Delta \der n_1 \eqinf n_3 : \funTs x U {s_1,s_2} T$ and
$\Delta \der \funTs x {U}{s_1,s_2}{T} = \funTs x {U'}{s_1',s_2'}{T'}$
which gives in particular $s_1 = s_1', s_2 = s_2'$,
and $\Delta \der U = U' : s_1$ by function type injectivity (Thm.~\ref{thm:funinj}). 
By induction hypothesis we can then deduce
$\Delta \der u_1 \eqchkt u_3 : U$, and therefore conclude
$\Delta \der n_1 u_1 \eqinf n_3 u_3 : \subst {u_1} x T$.

\item\noindent{\hskip-12 pt\em Case\enspace}
\begin{multline*}
  \ru{\Delta  \der U_1 \eqtyt U_2 : s_1 \quad
      \eext \Delta x U_1 \der T_1 \eqtyt T_2 : s_2
    }{\Delta \der \efunTs x {U_1} {s_1,s_2} {T_1} \eqty \efunTs x {U_2}{s_1,s_2}{T_2} : s_3
    }
\\[2ex]
  \ru{\Delta  \der U_2 \eqtyt U_3 : s_1 \quad
      \eext \Delta x {U_2} \der T_2 \eqtyt T_3 : s_2
    }{\Delta \der \efunTs x {U_2} {s_1,s_2} {T_2} \eqty \efunTs x {U_3}{s_1,s_2}{T_3} : s_2
    }
\end{multline*}

We get $\Delta \der U_1 \eqtyt U_3 : s_1 $ by transitivity.
To also get $\eext \Delta x {U_1} \der T_1 \eqtyt T_3 : s_2$
we need $\der \eext \Delta x {U_2} = \eext \Delta x {U_1}$,
but this stems from $\Delta \der U_1 \eqtyt U_2 : s_1$ by soundness of
algorithmic equality.\qed
\end{enumerate}

\begin{thm}
  \label{thm:aleqper}
  The algorithmic equality relations are PERs on well-typed expressions.
\end{thm}
\begin{proof}
  By Lemma~\ref{lem:aleqtrans} and an analogous proof of symmetry.
\end{proof}

\section{A Kripke Logical Relation for Completeness}
\label{sec:compl}

The only open issues in the meta-theory of \IITT are
completeness and termination of algorithmic equality.  
In parts, completeness has been
established in the last section already, namely, we have shown
injectivity and discrimination for type constructors.  What is missing
is injectivity and discrimination for neutrals, \eg, if 
$\Delta \der n\,u = n'\,u' : T'$ then necessarily 
$\Delta \der n = n' : \funT x U T$ and 
$\Delta \der u = u' : U$, plus $\Delta \der \subst u x T = T'$.
In untyped $\lambda$-calculus, this is an instance of Boehm's theorem
\cite{barendregt:lambdacalculus}.  We follow Coquand
\cite{coquand:conversion} and Harper and Pfenning
\cite{harperPfenning:equivalenceLF} and prove it by constructing a
second Kripke logical relation, $\Ceq$, for completeness 
which is very similar to the first one, $\Seq$, but at base types
additionally requires algorithmic equality to hold.  After proving the
fundamental lemma again, we know that definitionally equal terms are
also algorithmically so.  As a consequence, equality is decidable in
\IITT, and so is type checking.


\subsection{Another Kripke Logical Relation}

\renewcommand{\R}{\mathrel{:\Longleftrightarrow:}}

Again, by induction on
$A \in s$
we define two Kripke relations
\[
\begin{array}{r@{~}c@{~}l}
  \Delta \der A & \Ceq & A' : s
\\
  \Delta \der a & \Ceq & a' : A. 
\end{array}
\]
together with their respective closures $\Cleq$ and the generalization
to $\evar$.  This time, however, at base types we will additionally
require algorithmic equality to hold, more
precisely, the relation  $\Delta \der t \R t' : T$ which stands for
the conjunction of the propositions
\begin{iteMize}{$\bullet$}
\item $\Delta \der t : T$ and $\Delta \der t' : T$, and
\item $\Delta \der t \eqchkt t' : T$.
\end{iteMize}
Note that by soundness of algorithmic equality, $\R$ implies $\E$.

Again, we allow ourselves rule notation for the defining clauses of $\Ceq$.
\begin{gather*}
  \ru{\Delta \der N   \R  N' : s
    }{\Delta \der N \Ceq N' : s}
\qquad
  \ru{\Delta \der n  \R  n' : N
    }{\Delta \der n \Ceq  n' : N}
\qquad
  \rux{\der \Delta 
     }{\Delta \der s \Ceq  s : s'
     }{(s,s')} 
\end{gather*}
\begin{gather*}
 \rux{\lcol{\Delta \der U \Cleq U' : s_1} \\
      \forall \hatDelta \leq \Delta,~
      \hatDelta \der u \Cleq u' \evarof U \implies
      \hatDelta \der \subst u x T \Cleq \subst {u'}x{T'} : s_2 \\
      \Delta \der \efunTs x U {s_1,s_2} T  \E \efunTs x {U'}{s_1,s_2}{T'} : s_3
    }{\Delta \der \efunTs x U {s_1,s_2} T \Ceq \efunTs x {U'}{s_1,s_2}{T'} : s_3
    }{(s_1,s_2,s_3)}
\end{gather*}
\begin{gather*}
  \ru{\forall \hatDelta \leq \Delta,~
      \hatDelta \der u \Cleq u' \evarof U \implies
      \hatDelta \der f \eapp u  \Cleq f' \eapp u' : \subst {u}x{T}
          \\
      \Delta \der f \E f' : \efunTs x {U} {s,s'} {T}
    }{\Delta \der f \Ceq f' : \efunTs x {U} {s,s'} {T}}
\end{gather*}
\begin{gather*}
  \ru{
      \Delta \der \whnf t \Ceq \whnf t' : \whnf T \qquad
      \Delta \der t \E t' : T
    }{\Delta \der t \Cleq t' : T}
\end{gather*}
\begin{gather*}
  \ru{\resurrect\Delta \der a \Ceq a : A \qquad
      \resurrect{\Delta} \der a' \Ceq a' : A
    }{\Delta \der a \Ceq a' \div A}
\qquad
  \ru{\resurrect\Delta \der t \Cleq  t : T \qquad
      \resurrect{\Delta} \der t' \Cleq t' : T
    }{\Delta \der t \Cleq t' \div T}
\end{gather*}
This logical relation contains only well-typed
and definitionally equal terms.
It is symmetric, transitive, and closed under weakening and type
conversion.  The proofs are in analogy to those of 
Section~\ref{sec:sound}, which are relying on the fact that the underlying
relation $\E$ is a Kripke PER and closed under type conversion.  The
relation $\R$ underlying $\Ceq$ has the same properties, thanks to
soundness of algorithmic equality.

Note that in the definition of $\Delta \der f \Ceq f' : \efunT x U T$
we did not require $f$ and $f'$ to be algorithmically equal.  This
would hinder the proof of the fundamental theorem for $\Ceq$, 
since algorithmic
equality is not closed under application by definition---it will
follow from the fundamental theorem, though.  In the next lemma we
shall prove that $f$ and $f'$ \emph{are} algorithmically equal if they
are related by $\Ceq$.  The name \emph{Escape Lemma} was coined by
Jeffrey Sarnat \cite{schuermannSarnat:lics08}.

\begin{lem}[Escape from the logical relation]\label{lem:escape}
  Let $\Delta \der A \Ceq A' : s$
  \begin{enumerate}[\em(1)]
  \item \label{it:Sty}
        $\Delta \der A \eqty A'$.
  \item \label{it:Snf}
        If $\Delta \der t \Cleq t' : A$ then
        $\Delta \der t \eqchk t' : A$.
  \item \label{it:Sne}
        If $\Delta \der n \eqinf n' \evarof A$
        and $\Delta \der n = n' \evarof A$
        then $\Delta \der n \Ceq n' \evarof A$.
  \end{enumerate}
\end{lem}
\begin{cor}\label{cor:escape}
  Let $\Delta \der T \Cleq T' : s$
  \begin{enumerate}[\em(1)]
  \item \label{it:Cty}
        $\Delta \der T \eqtyt T'$.
  \item \label{it:Cnf}
        If $\Delta \der t \Cleq t' : T$ then
        $\Delta \der t \eqchkt t' : T$.
  \item \label{it:Cne}
        If $\Delta \der n \eqinft n' \evarof T$
        and $\Delta \der n = n' \evarof T$
        then $\Delta \der n \Cleq n' \evarof T$.
  \end{enumerate}
\end{cor}
The corollary is a direct, non-inductive consequence of the lemma, so
we can use it in the proof of the lemma, quoted as ``IH''.

\proof[Proof of the lemma]
Simultaneously by induction on $A \R A' : s$.

\begin{enumerate}[\hbox to8 pt{\hfill}] 

\item\noindent{\hskip-12 pt\em Case\enspace} $\Delta \der N \Ceq N' : s$.
  \begin{enumerate}[\hbox to8 pt{\hfill}] 

  \item\noindent{\hskip-12 pt\em Case\enspace} \ref{it:Sty}. $\Delta \der N \eqty N'$ by
  assumption.

  \item\noindent{\hskip-12 pt\em Case\enspace} \ref{it:Snf}.
  We have $\Delta \der \whnf t \eqinf \whnf t' : \_$,
  thus $\Delta \der t \eqchk t' : N$.

  \item\noindent{\hskip-12 pt\em Case\enspace} \ref{it:Sne}.

    First, consider $\evar = \noterased$.
    If $\Delta \der n = n' : N$
    and $\Delta \der n \eqinf n' : N$ then
    $\Delta \der n \eqchk n' : N$
    and trivially
    $\Delta \der n \Ceq n' : N$.

    Then, take $\evar = \erased$.
    Note that if $\resurrect\Delta \der n = n : N$
    and $\resurrect\Delta \der n \eqinf n : N$ then
    $\resurrect\Delta \der n \eqchk n : N$
    and $\resurrect\Delta \der n \Ceq n : N$.
    This implies that
    if $\Delta \der n = n' \div N$
    and $\Delta \der n \eqinf n' \div N$ then
    $\Delta \der n \eqchk n' \div N$
    and $\Delta \der n \Ceq n' \div N$.

  \end{enumerate}

\item\noindent{\hskip-12 pt\em Case\enspace} $\Delta \der s \Ceq s : s'$.
  \begin{enumerate}[\hbox to8 pt{\hfill}] 

  \item\noindent{\hskip-12 pt\em Case\enspace} \ref{it:Sty}. Clearly, $\Delta \der s \eqty s$.

  \item\noindent{\hskip-12 pt\em Case\enspace} \ref{it:Snf}.
  Let $\Delta \der T \Cleq T' : s$.
  Then $\Delta \der T \eqtyt T'$ by IH~\ref{it:Sty},
  thus $\Delta \der T \eqchk T' : s$

  \item\noindent{\hskip-12 pt\em Case\enspace} \ref{it:Sne}.  For $\evar = \noterased$
  let $\Delta \der N \eqinf N' : s$.
  By inversion, $\Delta \der N \eqinft N' : T$ for some $T$.
  Then $\Delta \der N \eqty N'$ and
  $\Delta \der N \Ceq N' : s$ by definition.

  Considering $\evar = \erased$, it is sufficient to observe that
  $\resurrect\Delta \der N \eqinf N : s$
  implies $\resurrect\Delta \der N \eqty N$ and
  $\resurrect\Delta \der N \Ceq N : s$ by definition.

  \end{enumerate}

\item\noindent{\hskip-12 pt\em Case\enspace}
  $\Delta \der \efunT x U T \Ceq \efunT x {U'}{T'} : s_3$.
  \begin{enumerate}[\hbox to8 pt{\hfill}] 

  \item\noindent{\hskip-12 pt\em Case\enspace} \ref{it:Sty}.  Similar to \ref{it:Snf}.
  \item\noindent{\hskip-12 pt\em Case\enspace} \ref{it:Snf}.
  By assumption,
  $\Delta \der t \Cleq t' : \efunT {x}{U}{T}$.  It is sufficient to show
  $\eext \Delta x { U} \der t \eapp x \eqchkt t' \eapp x :  T$.
  Since $\Delta \der U \Cleq U' : s_1$,
  which includes
  $\Delta \der U$, we have
  $\eext\Delta x U \der x = x \evarof U$.
  Since also
  $\eext\Delta x U \der x \eqinft x \evarof U$,
  we obtain
  $\eext\Delta x { U} \der t \eapp x
  \Cleq t' \eapp x : \whnf T$
  via IH~\ref{it:Sne},
  $\eext\Delta x U \der x \Ceq x \evarof U$.
  IH~\ref{it:Snf} then entails our goal.

  \item\noindent{\hskip-12 pt\em Case\enspace} \ref{it:Sne}.  First, the case for $\evar = \noterased$.
  We reuse variable $\evar$ for a different irrelevance marker.
  We have
  $\Delta \der n  \eqinf n' : \efunT x {U} {T}$.  Assume arbitrary
  $\hatDelta \R \hatDelta \leq \Delta$ and
  $\hatDelta  \der u  \Cleq u' \evarof U$,
  which yields
  $\hatDelta  \der u  = u' \evarof U$ and
  $\hatDelta  \der \subst u x T \Cleq \subst{u'}x{T} : \Set[i]$.
  In case $\evar = \noterased$ we have to apply IH~\ref{it:Snf}
  for
  $\hatDelta \der u \eqchk u' : \whnf{U}$.
  Otherwise, we obtain directly
  $\hatDelta  \der n \eapp u \eqinf n' \eapp u' : \whnf(\subst {u}x{T})$.
  By IH~\ref{it:Sne},
  $\hatDelta  \der n \eapp u \Ceq n' \eapp u' : \whnf(\subst {u}x{T})$.

   The case for $\evar = \erased$ proceeds analogously.\qed\medskip
   \end{enumerate}
\end{enumerate}

\noindent In analogy to $\cleq$ we extend $\Cleq$ to substitutions and define
the semantic validity judgements $\valic \Gamma$ and 
$\Gamma \valic t : T$ and $\Gamma \valic t = t' : T$ 
based on $\Cleq$.
Since by the escape lemma, $\Delta \der x \Cleq x : \Delta(x)$, we
have $\Gamma \der \sid \Cleq \sid : \Gamma$ for $\valic \Gamma$.
Finally, we reprove the fundamental theorem:
\begin{thm}[Fundamental theorem for $\Cleq$] 
  \label{thm:fundcompl} \bla
  \begin{enumerate}[\em(1)]
  \item If $\der \Gamma$ then $\valic \Gamma$.
  \item If $\Gamma \der t : T$ then $\Gamma \valic t : T$.
  \item If $\Gamma \der t = t' : T$ then
        $\Gamma \valic t = t' : T$.
  \end{enumerate}
\end{thm}

\subsection{Completeness and Decidability of Algorithmic Equality}

Derivations of algorithmic equality can now be obtained by escaping
from the logical relation.
\begin{thm}[Completeness of algorithmic equality]\label{thm:aleqcompl} 
  $\Gamma \der t = t' : T$ implies
  $\Gamma \der t \eqchkt t' :  T$.
\end{thm}
\begin{proof}
  Since $\Gamma \der \sid \Cleq \sid : \Gamma$,
  we have $\Gamma \der t \Cleq t' :  T$ by the fundamental theorem,
  and conclude with Lemma~\ref{lem:escape}.\ref{it:Snf}.
\end{proof}




Termination of algorithmic equality is a consequence of completeness.
When invoking the algorithmic equality check $\Delta \der t \eqchkt t'
: T$ on two well-typed expressions $\Delta \der t,t' : T$ we know by
completeness that $t$ and $t'$ are related to themselves, \ie,
$\Delta \der t \eqchkt t : T$ and $\Delta \der t' \eqchkt t' : T$.
This means that $t$, $t'$, and $T$ are weakly normalizing by the strategy
the equality algorithm implements:  reduce to weak head normal form
and recursively continue with the subterms.  Running the equality
check on $t$ and $t'$ performs, if successful, exactly the same
reductions, and if it fails, at most the same reductions in $t$,
$t'$, and $T$.  Hence, testing equality on well-typed terms always terminates.
This argument has been applied in previous work to untyped equality
\cite{abelCoquand:fundinf07}.  Here, we apply it to typed equality; it
is an alternative to Goguen's technique of proving termination for
typed equality from strong normalization
\cite{goguen:justifyingAlgorithms},
which, in our opinion, does not scale to dependently-typed equality.
\begin{lem}[Termination of algorithmic equality]\label{lem:term} 
  Let $\der \Delta$.
  \begin{enumerate}[\em(1)]
  \item Type equality.
    \begin{enumerate}[\em(a)]
    \item Let $\Delta \der A, A'$.
      If $\DD :: \Delta \der A \eqty A$ and 
      $\Delta \der A' \eqty A'$ then the query
      $\Delta \der A \eqty A'$ terminates.
    \item Let $\Delta \der T, T'$.
      If $\DD :: \Delta \der T \eqtyt T$ and 
      $\Delta \der T' \eqtyt T'$ then the query
      $\Delta \der T \eqtyt T'$ terminates.
    \end{enumerate}
  \item Structural equality.
    Let $\Delta \der n : T$ and $\Delta \der n' : T'$.
    \begin{enumerate}[\em(a)]
    \item If $\DD :: \Delta \der n \eqinf n : A$ and 
      $\Delta \der n' \eqinf n' : A'$ then the query
      $\Delta \der n \eqinf n' : \mathord{?}$ terminates.
      If successfully, it returns $A$ and we have
      $\Delta \der A = T = T' = A$.
    \item If $\DD :: \Delta \der n \eqinft n : T$ and 
      $\Delta \der n' \eqinft n' : T'$ then the query
      $\Delta \der n \eqinft n' : \mathord{?}$ terminates.
      If successfully, it returns $T$ and we have
      $\Delta \der T = T'$. 
    \end{enumerate}
  \item Type-directed equality.
    \begin{enumerate}[\em(a)]
    \item Let
    $\Delta \der t,t' : A$.
    If $\DD :: \Delta \der t \eqchk t : A$ and
    $\Delta \der t' \eqchk t' : A$ then the query
    $\Delta \der t \eqchk t' : A$ terminates.
    \item Let
    $\Delta \der t,t' : T$.
    If $\DD :: \Delta \der t \eqchkt t : T$ and
    $\Delta \der t' \eqchkt t' : T$ then the query
    $\Delta \der t \eqchkt t' : T$ terminates.
    \end{enumerate}
  \end{enumerate}
\end{lem}

\proof
  Simultaneously by induction on derivation $\DD$.
  \begin{enumerate}[(1)]

  \item Type equality.
    \begin{enumerate}[\hbox to8 pt{\hfill}]

    \item\noindent{\hskip-12 pt\em Case\enspace} $A = A' = s$.  The query $\Delta \der A \eqty A'$
    terminates successfully.

    \item\noindent{\hskip-12 pt\em Case\enspace} $A = \efunT x U T$ and $A' = \efunT x {U'}{T'}$.
    First, the query $\Delta \der U \eqtyt U'$ runs.  By induction
    hypothesis, it terminates.  If it fails, the
    whole query fails.  
    Otherwise, the query $\eext \Delta x U \der T \eqtyt T'$ is run.
    By induction hypothesis on $\eext \Delta x U \der T \eqtyt T$ and
    $\eext {\Delta} x {U'} \der T' \eqtyt T'$, the query terminates.

    \item\noindent{\hskip-12 pt\em Case\enspace} $A = N$ and $A' = N'$ neutral.
    By induction hypothesis on $\Delta \der N \eqinft N : T$ and
    $\Delta \der N' \eqinft N' : T'$, the query 
    $\Delta \der N \eqinft N' : \mathord?$ terminates.
    Hence, the query $\Delta \der N \eqty N'$ terminates.
     
    \item\noindent{\hskip-12 pt\em Case\enspace} Weak head normal forms $A,A'$ not covered by previous
    cases: the query $\Delta \der A \eqty A'$ fails immediately, since
    there is no applicable algorithmic type equality rule.

    \item\noindent{\hskip-12 pt\em Case\enspace} The query $\Delta \der T \eqtyt T'$ first invokes weak head
    normalization on $T$ and $T'$.  Both terminate since
    $\Delta \der T \eqtyt T$, which implies $T \evalsto A$, and
    analogously $T' \evalsto A'$ since $\Delta \der T' \eqtyt T'$ by
    assumption.   Then, the query $\Delta \der A \eqty A'$ is run,
    which terminates by induction hypothesis on $\Delta \der A \eqty
    A$ and $\Delta \der A' \eqty A'$.

    \end{enumerate}

  \item Structural equality. 
    \begin{enumerate}[\hbox to8 pt{\hfill}]
      
    \item\noindent{\hskip-12 pt\em Case\enspace} $n = n' = x$.  The query 
    $\Delta \der n \eqinft n' : \mathord?$ terminates successfully,
    returning type $\Delta(x)$.  Since $\der \Delta$, by
    inversion (Lemma~\ref{lem:inv}) $\Delta \der T = T' = \Delta(x)$.

    \item\noindent{\hskip-12 pt\em Case\enspace} Neutral relevant application for 
    $\Delta \der n \app u : T_0$ and $\Delta \der n' \app u' : T'_0$.
\begin{multline*}  
  \ru{\Delta \der n \eqinf n : \funT x U T \qquad
      \Delta \der u \eqchkt u : U
    }{\Delta \der n \app u \eqinft n \app u : \subst u x T}
 \\
  \ru{\Delta \der n' \eqinf n' : \funT x {U'} {T'} \qquad
      \Delta \der u' \eqchkt u' : U'
    }{\Delta \der n' \app u' \eqinft n' \app u' : \subst {u'} x {T'}}
\end{multline*}
    The query $\Delta \der n \app u \eqinft n' \app u' : \mathord?$
    first invokes query $\Delta \der n \eqinf n' : \mathord?$.
    By induction hypothesis on $\Delta \der n \eqinf n : \funT x U T$
    and $\Delta \der n' \eqinf n' : \funT x {U'}{T'}$ the query
    terminates.  If it fails the whole query fails.  Otherwise it
    returns a type $A$ in weak head normal form, which is
    identical to $\funT x U T$      
    by uniqueness of inferred types (Lemma~\ref{lem:uniqinfer}).
    Further, $\Delta \der \funT x U T = \funT x {U'}{T'}$, and by
    function type injectivity (Thm.~\ref{thm:funinj}), 
    $\Delta \der U = U'$ and $\cext \Delta x U \der T = T'$.
    Thus, we can invoke the induction hypothesis on 
    $\Delta \der u \eqchkt u : U$ and $\Delta \der u' \eqchkt u' : U$
    (cast from $\Delta \der u' \eqchkt u' : U'$, Lemma~\ref{lem:tyconvalg})
    to infer that the second subquery $\Delta \der u \eqchkt u' : U$
    terminates.  If this one is successful, then by soundness of
    algorithmic equality, $\Delta \der u = u' : U$, which implies
    $\Delta \der \subst u x T = \subst {u'} x {T'}$.

    \item\noindent{\hskip-12 pt\em Case\enspace} Neutral irrelevant application with typing
\[ 
   \ru{\Delta \der n : \erfunT x {U_1} {T_1} \qquad
       \Delta \der u \erof U_1
     }{\Delta \der n \erapp u : \subst u x {T_1}}
\qquad 
   \ru{\Delta \der n' : \erfunT x {U'_1} {T'_1} \qquad
       \Delta \der u' \erof {U'_1}
     }{\Delta \der n' \erapp u' : \subst {u'} x {T'_1}}
\]
   and algorithmic self-equality 
\[ 
   \ru{\Delta \der n \eqinf n : \erfunT x U T
     }{\Delta \der n \erapp u \eqinft n \erapp u : \subst u x T}
\qquad 
   \ru{\Delta \der n' \eqinf n' : \erfunT x {U'} {T'}
     }{\Delta \der n' \erapp u' \eqinft n' \erapp u' : \subst {u'} x {T'}}
\]
    The query $\Delta \der n \erapp u \eqinft n' \erapp u' : \mathord?$
    invokes query $\Delta \der n \eqinf n' : \mathord?$, which
    terminates by induction hypothesis.  
    If successfully, then 
    $\Delta \der \erfunT x {U_1}{T_1} \erfunT x U T = \erfunT x
    {U'}{T'} = \erfunT x {U_1'}{T_1'}$.
    By function type injectivity, $\Delta \der U_1 = U = U' = U_1'$
    and $\erext \Delta x U \der T_1 = T = T' = T'_1$.  By conversion 
    $\Delta \der u = u' \erof U$, thus, 
    $\Delta \der \subst u x {T_1} = \subst u x T = \subst {u'} x {T'}
    = \subst {u'} x {T'_1}$.

    \item\noindent{\hskip-12 pt\em Case\enspace} In all other cases, the query $\Delta \der n \eqinft n'
    : \mathord?$ fails immediately.

    \item\noindent{\hskip-12 pt\em Case\enspace} The query $\Delta \der n \eqinf n : \mathord?$ spawns
    subquery $\Delta \der n \eqinft n' : \mathord?$ which terminates by
    induction hypothesis on $\Delta \der n \eqinft n : T$ and 
    $\Delta \der n' \eqinft n' : T'$.  If successfully, it returns
    type $T$, and since $T \evalsto A$, the original query also
    terminates, returning $A$.

    \end{enumerate}

  \item Type-directed equality.
    \begin{enumerate}[\hbox to8 pt{\hfill}]
      
    \item\noindent{\hskip-12 pt\em Case\enspace} Function type $\Delta \der t,t' : \efunT x U T$.  The query 
    $\Delta \der t \eqchk t' : \efunT x U T$ spawns subquery 
    $\eext \Delta x U \der t \eapp x \eqchkt t' \eapp : T$.  Since   
    $\eext \Delta x U \der t \eapp x, t' \eapp x : T$ and the subquery
    terminates by induction hypothesis on 
    $\eext \Delta x U \der t \eapp x \eqchkt t \eapp x : T$ and
    $\eext \Delta x U \der t' \eapp x \eqchkt t' \eapp x : T$.

    \item\noindent{\hskip-12 pt\em Case\enspace} Sort $\Delta \der T,T' : s$.  The query $\Delta \der T \eqchk
    T' : s$ calls $\Delta \der T \eqtyt T'$, which terminates by
    induction hypothesis on $\Delta \der T \eqtyt T$ and $\Delta \der
    T' \eqtyt T'$.

    \item\noindent{\hskip-12 pt\em Case\enspace} Neutral type $N$.
\[
  \ru{t \evalsto n \qquad
      \Delta \der n \eqinft n : T
    }{\Delta \der t \eqchk t : N}
\qquad
   \ru{t' \evalsto n' \qquad
       \Delta \der n' \eqinft n' : T'
     }{\Delta \der t' \eqchk t' : N} 
\]
    The query $\Delta \der t \eqchk t' : N$ first
    weak head normalizes $t$ and $t'$.  By assumption,
    $t \evalsto n$ and $t' \evalsto n'$, so this terminates.
    The subquery $\Delta \der n \eqinft n' : \mathord?$ terminates by
    induction hypothesis.  Thus, the whole query terminates.  
   
    \item\noindent{\hskip-12 pt\em Case\enspace} If $A$ is neither a function type, a sort, or a neutral
    type, the query $\Delta \der t \eqchk t' : A$ fails immediately.

    \item\noindent{\hskip-12 pt\em Case\enspace} The query $\Delta \der t \eqchkt t' : T$ first weak head
    normalizes $T$ which terminates since $T \evalsto A$ by
    assumption.  Then it calls $\Delta \der t \eqchk t'  : A$ which
    terminates by induction hypothesis.\qed\medskip
    \end{enumerate}
  \end{enumerate}

\begin{thm}\label{thm:term}
  If $\Delta \der t : T$ and
  $\Delta \der t' : T$ then the query
  $\Delta \der t \eqchkt t' : T$ terminates.
\end{thm}
\begin{proof}
  From the lemma by completeness of algorithmic equality.
\end{proof}

Thus we have shown that algorithmic equality is correct, \ie, sound, complete, and
terminating. Together, this entails decidability of equality in \IITT.

\begin{thm}[Decidability of \IITT]
  \bla
  \begin{enumerate}[\em(1)]
  \item $\Gamma \der t = t' : T$ is decidable.
  \item $\Gamma \der t : T$ is decidable.
  \end{enumerate}
\end{thm}
\begin{proof}
  Decidability of equality follows from soundness
  (Thm.~\ref{thm:aleqsound}), completeness (Thm.~\ref{thm:aleqcompl}), and
  termination~(Thm.~\ref{thm:term}).
  Decidability of typing follows from decidability of type conversion,
  weak head normalization, and function type injectivity, using 
  inversion (Lemma~\ref{lem:inv}) 
  on typing derivations.  Any reasonable type
  inference algorithm will do.
\end{proof}

\section{Extensions}
\label{sec:ext}

\paradot{Data types and recursion}
The semantics of \IITT is ready to cope with inductive data types like the
natural numbers and the associated recursion principles.  Recursion
into types, aka known as large elimination, is also accounted for
since we have universes and a semantics which does not erase
dependencies (unlike Pfenning's model \cite{pfenning:intextirr}).

\paradot{Types with extensionality principles}
One purpose of having a typed equality algorithm is to handle
$\eta$-laws that are not connected to the shape of the expression
(like $\eta$-contraction for functions) but to the shape of the type
only.  Typically these are types $T$ with at most one inhabitant, \ie, the
empty type, the unit type, singleton types or
propositions.\footnote{Some care is necessary for the type of Leibniz
  equality \cite{abel:nbe09,werner:lmcs08}.}
For such $T$ we have the $\eta$-law
\[
  \ru{\Gamma \der t,t' : T
    }{\Gamma \der t = t' : T}
\]
which can only be checked in the presence of type $T$.
Realizing such $\eta$-laws gives additional ``proof'' irrelevance
which is not covered by Pfenning's irrelevant quantification
$\erfunT x U T$.

\paradot{Internal erasure}
Terms $u \div U$ in irrelevant position are only there to please the
type checker, they are ignored during equality
checking. This can be inferred from the substitution principle:  If
$\erext \Gamma x U \der T$ and $\Gamma \der u, u' \div U$, then
$\Gamma \der \subst u x T = \subst {u'} x T$; the type $T$ has
the same shape regardless of $u,u'$.
Hence, terms like $u$ serve the sole purpose to prove some proposition
and could be replaced by a dummy $\tirr$ immediately after
type-checking.  

Internal erasure can be realized by making $\Gamma \der t
\div T$ a judgement (as opposed to just a notation for $\resurrect\Gamma
\der t : T$) and adding the rule
\[
\ru{\Gamma \der t \div T
  }{\Gamma \der \tirr \div T} .
\]
The rule states that if there is already a proof $t$ of $T$, then $\tirr$ is a
new proof of $T$.  This preserves provability while erasing the proof
terms.  Conservativity of this rule can be proven as in joint work of
the author with Coquand and Pagano \cite{abelCoquandPagano:lmcs11}.

\section{Conclusions}
\label{sec:concl}

We have extended Pfenning's notion of irrelevance to a type theory \IITT
with universes that accommodates types defined by recursion. 
We have constructed a Kripke model $\Seq$ that shows soundness of \IITT,
yielding normalization, subject reduction and consistency,
plus syntactical properties of the judgements of \IITT.
A second Kripke logical relation $\Ceq$ has proven correctness
of algorithmic equality and, thus, decidability of \IITT.  

Integrating irrelevance and data types in dependent type theory
does not seem without challenges.  We have succeeded to treat
Pfenning's notion of irrelevance, but our proof does not scale
directly to \emph{parametric} function types, a stronger notion of
irrelevant function types called implicit quantification by Miquel
\cite{miquel:PhD}.\footnote{A function argument is parametric if it is
irrelevant for computing the function result while the type of the
result may depend on it.  In Pfenning's notion, the argument must also
be irrelevant in the type.}  Two more type theories build on Miquel's
calculus \cite{miquel:tlca01},  Barras and Bernardo's \ICCstar
\cite{barrasBernardo:fossacs08} and Mishra-Linger and Sheard's
\emph{Erasure Pure Type Systems} (EPTS)
\cite{mishraLingerSheard:fossacs08}, but none has offered a satisfying
account of large eliminations yet.  Miquel's model
\cite{miquel:lics00} features data types only as impredicative
encodings.   For irrelevant, parametric, and recursive functions to
coexist it seems like three different function types are necessary, \eg, in
the style of Pfenning's \emph{irrelevance, extensionality and intensionality}.
We would like to solve this puzzle in future work, not least to
implement high-performance languages with dependent types.

\paradot{Acknowledgments} The first author thanks Bruno Barras, Bruno Bernardo,
Thierry Coquand, Dan Doel, Hugo Herbelin, Conor McBride, Ulf
Norell, and Jason Reed for discussions on irrelevance in type theory.
Work on a
previous paper has been carried out while he was invited researcher at
PPS, Paris, in the INRIA $\pi r^2$ team headed by Pierre-Louis Curien
and Hugo Herbelin.  
The second author acknowledges financial support by the \'Ecole Normale
Superi\'eure de Paris for his internship at the
Ludwig-Maximilians-Universit\"at M\"unchen from May to September 2011.
We thank the two anonymous referees, who suggested changes and
examples which significantly improved the presentation, 
and the patience of the editors waiting for our revisions. 


\bibliographystyle{alpha}
\bibliography{auto-types10}\vspace{-50 pt}

\newcommand{\etalchar}[1]{$^{#1}$}
\begin{thebibliography}{CAB{\etalchar{+}}86}

\bibitem[AB04]{awodeyBauer:propositionsAsTypes}
Steven Awodey and Andrej Bauer.
\newblock Propositions as [{T}ypes].
\newblock {\em Journal of Logic and Computation}, 14(4):447--471, 2004.

\bibitem[Abe09]{abel:nbe09}
Andreas Abel.
\newblock Extensional normalization in the logical framework with proof
  irrelevant equality.
\newblock In Olivier Danvy, editor, {\em Workshop on Normalization by
  Evaluation, affiliated to LiCS 2009, Los Angeles, 15 August 2009}, 2009.

\bibitem[Abe11]{abel:fossacs11}
Andreas Abel.
\newblock Irrelevance in type theory with a heterogeneous equality judgement.
\newblock In Martin Hofmann, editor, {\em Foundations of Software Science and
  Computational Structures, 14th International Conference, FOSSACS 2011, Held
  as Part of the Joint European Conferences on Theory and Practice of Software,
  ETAPS 2011, Saarbr{\"{u}}cken, Germany, March 26 - April 3, 2011.
  Proceedings}, volume 6604 of {\em Lecture Notes in Computer Science}, pages
  57--71. Springer-Verlag, 2011.

\bibitem[AC07]{abelCoquand:fundinf07}
Andreas Abel and Thierry Coquand.
\newblock Untyped algorithmic equality for {Martin-L\"of's} logical framework
  with surjective pairs.
\newblock {\em Fundamenta Informaticae}, 77(4):345--395, 2007.
\newblock {TLCA'05} special issue.

\bibitem[ACD07]{abelCoquandDybjer:lics07}
Andreas Abel, Thierry Coquand, and Peter Dybjer.
\newblock Normalization by evaluation for {Martin-L\"of Type Theory} with typed
  equality judgements.
\newblock In {\em 22nd IEEE Symposium on Logic in Computer Science (LICS 2007),
  10-12 July 2007, Wroclaw, Poland, Proceedings}, pages 3--12. IEEE Computer
  Society Press, 2007.

\bibitem[ACD08]{abelCoquandDybjer:mpc08}
Andreas Abel, Thierry Coquand, and Peter Dybjer.
\newblock Verifying a semantic $\beta\eta$-conversion test for {Martin-L\"of}
  type theory.
\newblock In Philippe Audebaud and Christine Paulin-Mohring, editors, {\em
  Mathematics of Program Construction, 9th International Conference, MPC 2008,
  Marseille, France, July 15-18, 2008. Proceedings}, volume 5133 of {\em
  Lecture Notes in Computer Science}, pages 29--56. Springer-Verlag, 2008.

\bibitem[ACP11]{abelCoquandPagano:lmcs11}
Andreas Abel, Thierry Coquand, and Miguel Pagano.
\newblock A modular type-checking algorithm for type theory with singleton
  types and proof irrelevance.
\newblock {\em Logical Methods in Computer Science}, 7(2:4):1--57, May 2011.

\bibitem[Ada06]{adams:jfp06}
Robin Adams.
\newblock Pure type systems with judgemental equality.
\newblock {\em Journal of Functional Programming}, 16(2):219--246, 2006.

\bibitem[All87]{allen:PhD}
Stuart Allen.
\newblock {\em A Non-Type-Theoretic Semantics for Type-Theoretic Language}.
\newblock PhD thesis, Cornell University, 1987.

\bibitem[Ama08]{DBLP:conf/fossacs/2008}
Roberto~M. Amadio, editor.
\newblock {\em Foundations of Software Science and Computational Structures,
  11th International Conference, FOSSACS 2008, Held as Part of the Joint
  European Conferences on Theory and Practice of Software, ETAPS 2008,
  Budapest, Hungary, March 29 - April 6, 2008. Proceedings}, volume 4962 of
  {\em Lecture Notes in Computer Science}. Springer-Verlag, 2008.

\bibitem[Aug99]{augustsson:cayenne}
Lennart Augustsson.
\newblock Cayenne - a language with dependent types.
\newblock In {\em Proceedings of the third ACM SIGPLAN International Conference
  on Functional Programming (ICFP '98), Baltimore, Maryland, USA, September
  27-29, 1998}, volume~34 of {\em SIGPLAN Notices}, pages 239--250. ACM Press,
  1999.

\bibitem[Bar84]{barendregt:lambdacalculus}
Henk Barendregt.
\newblock {\em The Lambda Calculus: Its Syntax and Semantics}.
\newblock North Holland, Amsterdam, 1984.

\bibitem[BB08]{barrasBernardo:fossacs08}
Bruno Barras and Bruno Bernardo.
\newblock The implicit calculus of constructions as a programming language with
  dependent types.
\newblock In Amadio \cite{DBLP:conf/fossacs/2008}, pages 365--379.

\bibitem[BDN09]{boveDybjerNorell:tphols09}
Ana Bove, Peter Dybjer, and Ulf Norell.
\newblock A brief overview of {Agda} - a functional language with dependent
  types.
\newblock In Stefan Berghofer, Tobias Nipkow, Christian Urban, and Makarius
  Wenzel, editors, {\em Theorem Proving in Higher Order Logics, 22nd
  International Conference, TPHOLs 2009, Munich, Germany, August 17-20, 2009.
  Proceedings}, volume 5674 of {\em Lecture Notes in Computer Science}, pages
  73--78. Springer-Verlag, 2009.

\bibitem[CAB{\etalchar{+}}86]{constable:nuprl}
Robert~L. Constable, Stuart~F. Allen, Mark Bromley, Rance Cleaveland, J.~F.
  Cremer, Robert~W. Harper, Douglas~J. Howe, Todd~B. Knoblock, Nax~P. Mendler,
  Prakash Panangaden, James~T. Sasaki, and Scott~F. Smith.
\newblock {\em Implementing mathematics with the Nuprl proof development
  system}.
\newblock Prentice Hall, 1986.

\bibitem[Coq91]{coquand:conversion}
Thierry Coquand.
\newblock An algorithm for testing conversion in type theory.
\newblock In G.~Huet and G.~Plotkin, editors, {\em Logical Frameworks}, pages
  255--279. Cambridge University Press, 1991.

\bibitem[Coq96]{coquand:type}
Thierry Coquand.
\newblock An algorithm for type-checking dependent types.
\newblock In {\em Mathematics of Program Construction. Selected Papers from the
  Third International Conference on the Mathematics of Program Construction
  (July 17--21, 1995, Kloster Irsee, Germany)}, volume~26 of {\em Science of
  Computer Programming}, pages 167--177. Elsevier, May 1996.

\bibitem[Gog94]{goguen:PhD}
Healfdene Goguen.
\newblock {\em A Typed Operational Semantics for Type Theory}.
\newblock PhD thesis, University of Edinburgh, August 1994.
\newblock Available as LFCS Report ECS-LFCS-94-304.

\bibitem[Gog00]{goguen:types00}
Healfdene Goguen.
\newblock A {K}ripke-style model for the admissibility of structural rules.
\newblock In Paul Callaghan, Zhaohui Luo, James McKinna, and Robert Pollack,
  editors, {\em Types for Proofs and Programs, International Workshop, TYPES
  2000, Durham, UK, December 8-12, 2000, Selected Papers}, volume 2277 of {\em
  Lecture Notes in Computer Science}, pages 112--124. Springer-Verlag, 2000.

\bibitem[Gog05]{goguen:justifyingAlgorithms}
Healfdene Goguen.
\newblock Justifying algorithms for $\beta\eta$ conversion.
\newblock In Vladimiro Sassone, editor, {\em Foundations of Software Science
  and Computational Structures, 8th International Conference, FoSSaCS 2005,
  Held as Part of the Joint European Conferences on Theory and Practice of
  Software, ETAPS 2005, Edinburgh, UK, April 4-8, 2005, Proceedings}, volume
  3441 of {\em Lecture Notes in Computer Science}, pages 410--424.
  Springer-Verlag, 2005.

\bibitem[HP05]{harperPfenning:equivalenceLF}
Robert Harper and Frank Pfenning.
\newblock On equivalence and canonical forms in the {LF} type theory.
\newblock {\em ACM Transactions on Computational Logic}, 6(1):61--101, 2005.

\bibitem[INR10]{inria:coq83}
INRIA.
\newblock {\em The Coq Proof Assistant Reference Manual}.
\newblock INRIA, version 8.3 edition, 2010.

\bibitem[Let02]{letouzey:types02}
Pierre Letouzey.
\newblock A new extraction for {Coq}.
\newblock In Herman Geuvers and Freek Wiedijk, editors, {\em Types for Proofs
  and Programs, Second International Workshop, TYPES 2002, Berg en Dal, The
  Netherlands, April 24-28, 2002, Selected Papers}, volume 2646 of {\em Lecture
  Notes in Computer Science}, pages 200--219. Springer-Verlag, 2002.

\bibitem[Miq00]{miquel:lics00}
Alexandre Miquel.
\newblock A model for impredicative type systems, universes, intersection types
  and subtyping.
\newblock In {\em 15th IEEE Symposium on Logic in Computer Science (LICS 2000),
  26-29 June 2000, Santa Barbara, California, USA, Proceedings}, pages 18--29,
  2000.

\bibitem[Miq01a]{miquel:tlca01}
Alexandre Miquel.
\newblock The implicit calculus of constructions.
\newblock In Samson Abramsky, editor, {\em Typed Lambda Calculi and
  Applications, 5th International Conference, TLCA 2001, Krakow, Poland, May
  2-5, 2001, Proceedings}, volume 2044 of {\em Lecture Notes in Computer
  Science}, pages 344--359. Springer-Verlag, 2001.

\bibitem[Miq01b]{miquel:PhD}
Alexandre Miquel.
\newblock {\em Le {C}alcul des {C}onstructions implicite: syntaxe et
  s\'emantique.}
\newblock PhD thesis, Universit\'e Paris 7, December 2001.

\bibitem[ML08]{mishraLinger:PhD}
Richard~Nathan Mishra-Linger.
\newblock {\em Irrelevance, Polymorphism, and Erasure in Type Theory}.
\newblock PhD thesis, Portland State University, 2008.

\bibitem[MLS08]{mishraLingerSheard:fossacs08}
Nathan Mishra-Linger and Tim Sheard.
\newblock Erasure and polymorphism in pure type systems.
\newblock In Amadio \cite{DBLP:conf/fossacs/2008}, pages 350--364.

\bibitem[MM04]{mcBrideMcKinna:view}
Conor McBride and James McKinna.
\newblock The view from the left.
\newblock {\em Journal of Functional Programming}, 14(1):69--111, 2004.

\bibitem[Pfe01]{pfenning:intextirr}
Frank Pfenning.
\newblock Intensionality, extensionality, and proof irrelevance in modal type
  theory.
\newblock In {\em 16th IEEE Symposium on Logic in Computer Science (LICS 2001),
  16-19 June 2001, Boston University, USA, Proceedings}. IEEE Computer Society
  Press, 2001.

\bibitem[PMW93]{paulinWerner:jsc93}
Christine Paulin-Mohring and Benjamin Werner.
\newblock Synthesis of {ML} programs in the system {Coq}.
\newblock {\em Journal of Symbolic Computation}, 15(5/6):607--640, 1993.

\bibitem[Ree02]{reed:thesis}
Jason Reed.
\newblock Proof irrelevance and strict definitions in a logical framework,
  2002.
\newblock Senior Thesis, published as Carnegie-Mellon University technical
  report CMU-CS-02-153.

\bibitem[Ree03]{reed:tphols03}
Jason Reed.
\newblock Extending higher-order unification to support proof irrelevance.
\newblock In David~A. Basin and Burkhart Wolff, editors, {\em Theorem Proving
  in Higher Order Logics, 16th International Conference, TPHOLs 2003, Rom,
  Italy, September 8-12, 2003, Proceedings}, volume 2758 of {\em Lecture Notes
  in Computer Science}, pages 238--252. Springer-Verlag, 2003.

\bibitem[SS08]{schuermannSarnat:lics08}
Carsten Sch{\"u}rmann and Jeffrey Sarnat.
\newblock Structural logical relations.
\newblock In Frank Pfenning, editor, {\em Proceedings of the Twenty-Third
  Annual IEEE Symposium on Logic in Computer Science, LICS 2008, 24-27 June
  2008, Pittsburgh, PA, USA}, pages 69--80. IEEE Computer Society Press, 2008.

\bibitem[VC02]{vanderwaartCrary:lfm02}
Joseph~C. Vanderwaart and Karl Crary.
\newblock A simplified account of the metatheory of {Linear LF}.
\newblock In {\em Third International Workshop on Logical Frameworks and
  Metalanguages (LFM 2002), FLoC'02 affiliated workshop, Copenhagen, Denmark},
  2002.
\newblock An extended version appeared as CMU Technical Report CMU-CS-01-154.

\bibitem[Wer08]{werner:lmcs08}
Benjamin Werner.
\newblock On the strength of proof-irrelevant type theories.
\newblock {\em Logical Methods in Computer Science}, 4(3), 2008.

\end{thebibliography}
\end{document}